\PassOptionsToPackage{table}{xcolor}

\documentclass[acmtog,nonacm,balance=false]{acmart}
\AtBeginDocument{%
  }

\newtheorem{proposition}{Proposition}

\newtheorem{corollary}{Corollary}

\newtheorem{definition}{Definition}

\newcommand{\changed}[2]{{\label{#2}#1}} %

\newcommand{\newchanged}[2]{\textcolor{black}{\label{#2}#1}} %
\newenvironment{newchangedblock}{\color{black}}{} %

\newcommand{\surface}{\Omega}
\newcommand{\point}{\mathbf{p}}
\newcommand{\otherpoint}{\mathbf{q}}
\newcommand{\smoothpt}{\mathbf{x}}
\newcommand{\sdf}{\phi}
\newcommand{\sdfsample}{s}
\newcommand{\growtopoints}{P}
\newcommand{\qmin}{\mathbf{q}^\star}
\newcommand{\rmin}{s^\star}
\newcommand{\diagonal}{\delta}
\newcommand{\subdivisiondepth}{\tau}
\DeclareMathOperator{\distance}{dist}
\DeclareMathOperator{\sign}{sign}

\definecolor{solidblue}{HTML}{447ebc}
\definecolor{solidred}{HTML}{ed551f}
\definecolor{lightblue}{HTML}{abcce1}
\definecolor{lightred}{HTML}{f59275}

\usepackage{layouts}
\makeatletter 
\newcommand{\layoutdetails}{%
\begin{tabular}{ll}
 \texttt{\textbackslash{textwidth}} & \printinunitsof{in}\prntlen{\textwidth} \\
\texttt{\textbackslash{linewidth}} & \printinunitsof{in}\prntlen{\linewidth} \\
Main text font &  \f@size pt \f@family \\
\sffamily \small Caption text font &  \sffamily \small \f@size pt \f@family \\
\end{tabular}%
}
\makeatother 
\usepackage{enumitem}
\usepackage{placeins}
\usepackage{wrapfig}
\usepackage{listings}
\makeatletter
\let\ftype@table\ftype@figure
\makeatother

\acmSubmissionID{}

\begin{document}

\title{\newchanged{Refining and Repairing Signed Distance Data With Global Consistency}{title-change}}

\author{Letao Chen}
\affiliation{%
  \institution{University of Southern California}
  \country{USA}
}
\author{Sanju Mupparaju}
\affiliation{%
  \institution{Massachusetts Institute of Technology}
  \country{USA}
}
\author{Christopher Batty}
\affiliation{%
  \institution{University of Waterloo}
  \country{Canada}
}
\author{Silvia Sell\'an}
\affiliation{%
  \institution{Columbia University}
  \country{USA}
}
\affiliation{%
  \institution{Massachusetts Institute of Technology}
  \country{USA}
}
\authornote{Joint last authors}
\author{Oded Stein}
\affiliation{
  \institution{Technion}
  \country{Israel}
}
\affiliation{
  \institution{University of Southern California}
  \country{USA}
}
\authornotemark[1]
\renewcommand{\shortauthors}{Chen et al.}

\begin{abstract}
We propose a method to interpolate \newchanged{exact}{} Signed Distance Function (SDF) data from a discrete set of samples.
Unlike prior work, our approach ensures that the new SDF data values are fully consistent with the input and each other, such that the augmented data still corresponds to a geometrically realizable surface.
We express the theoretical properties of \newchanged{exact}{} SDFs as hard geometric constraints,
and construct an efficient greedy algorithm for consistent SDF interpolation that is made even faster with powerful parallelized GPU preprocessing.
We exemplify the usefulness of our method by evaluating it on
three practical applications:
global SDF refinement, in which the SDF data is upsampled without knowledge of the ground truth; 
mesh reconstruction, where our method can reconstruct
highly detailed surfaces using global information from coarse input SDFs;
and repair of pseudo-SDFs, which result from many pipelines such as CSG Boolean operations and must be turned into valid SDFs for downstream processing tasks.
Our refined SDFs are guaranteed to be consistent with the input, where previous methods have no such guarantee.
\end{abstract}


\begin{teaserfigure}
  \includegraphics[width=\textwidth]{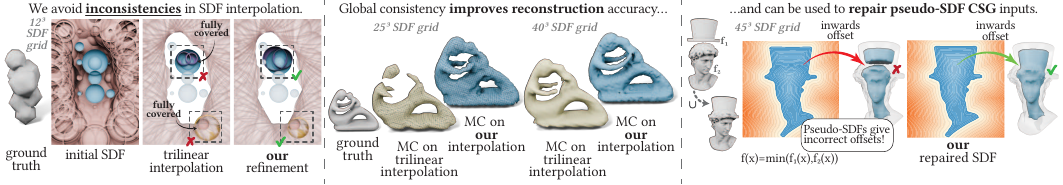}
  \vspace{-0.7cm}
  \caption{We interpolate signed distance to arbitrary spatial positions \newchanged{so that}{} the newly produced values are consistent with the discrete input SDF.
  This can be used for refining a coarse SDF \emph{(left, visualized as SDF spheres)}, where \newchanged{some}{} previous methods \newchanged{can}{} produce \newchanged{invalid spheres that either fully contain an input sphere (yellow) or are themselves fully contained by one (purple)}{};
  \newchanged{the marching cubes reconstruction result using our refined data is more accurate than those obtained by other interpolation methods \emph{(center)}}{};
  and we can repair pseudo-SDFs from CSG inputs (\emph{right}) that otherwise lead to downstream errors.}
  \label{fig:teaser}
\end{teaserfigure}

\maketitle

\pagenumbering{arabic}
\setcounter{page}{1}

\begin{figure}[b]
\vspace{-0.2cm}
\centering
\includegraphics[width=3.37in]{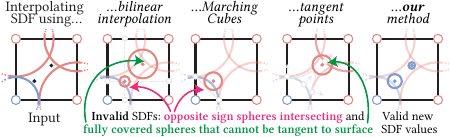}
\vspace{-0.7cm}
\caption{Inconsistencies between distance values are most recognizable when SDFs are visualized using spheres \cite{rfts}. Prior interpolation strategies produce new SDF values that contradict the inputs, in the form of intersecting opposite-sign spheres (violating eikonality) and fully covered spheres (violating the \emph{closest point property} \cite{Marschner2023}).}
\label{fig:fig2}
\end{figure}

\section{Introduction}
\label{sec:intro}

\emph{Signed Distance Functions} (SDFs) represent solid shapes by measuring the \newchanged{Euclidean}{} distance from any point in space to the solid's surface and assigning a negative sign to those points inside the solid.
The mathematical properties and flexibility of SDFs have made them the representation of choice in applications from industrial design to 3D generative AI and many others.
\newchanged{In some cases, one does not have ready access to a continuous source SDF, but only a finite set of samples, which we call a discrete SDF. Discrete SDFs are regularly used as intermediate or output representations in simulation, geometric modeling, and geometric processing pipelines. For example, particle and marker level-set methods \cite{mihalef2007textured,ENRIGHT2002} combine surface particles with discretely sampled level-set fields to track deforming surfaces in liquid animation. Implicit shape editing has often used discrete SDFs \cite{Museth2002,houston2006hierarchical}, notably for Boolean operations. In production-oriented pipelines, mesh geometry or particle configurations can be converted into standalone signed-distance volumes, as supported by the popular OpenVDB library \cite{museth2013vdb,museth2021nanovdb}, allowing the resulting grid data to be stored and processed without necessarily retaining the original surface format. Discrete SDFs are also popular in the context of collision detection among rigid and deformable bodies \cite{guendelman2003nonconvex}, including for complex polygon soup meshes \cite{Xu2014SDF}. In recent years, this approach has been extended to \emph{continuous} collision detection (CCD) on such sampled SDF grids for point-SDF \cite{xu2014continuous}, triangle-SDF \cite{pelletier2025real}, and SDF-SDF \cite{nielsen2026alternating} queries.}{MOTIVATION}

\newchanged{Given such a discrete SDF (and no access to its corresponding ground truth continuous counterpart), a critically important task is to \emph{interpolate} SDF values at new spatial query points. Like recent work on discrete SDF reconstruction (e.g., \cite{rfta,kohlbrenner2025polyhedral,supersdf}), we focus on the setting where the discrete SDF contains exact distances to the (unknown) surface. Unfortunately, even then, generic polynomial interpolation methods can produce outputs that violate the most basic mathematical properties of SDFs, possibly causing failure or worsened efficiency of downstream algorithms (e.g., unit-norm gradients or \emph{eikonality} prevents \changed{sphere tracing}{sphere-tracing} from overshooting the surface; the distance property enables sphere tracing and closest point projections with fewer iterations).}{}
Even brute force \emph{redistancing} strategies, which inefficiently reconstruct an explicit mesh from the input samples and then compute distances to it, \newchanged{can}{} produce values inconsistent with the input (see \autoref{fig:fig2}).
The main challenge \newchanged{of exact SDF interpolation}{} lies in the non-convex and highly global nature of the problem: any new interpolated SDF value must be consistent not only with the existing samples closest to it, but also with those arbitrarily far from it.

In this paper, we introduce this \newchanged{problem}{} to the geometry processing and computer graphics communities and propose \newchanged{an algorithm}{} to interpolate \newchanged{exact}{} SDF data while ensuring consistency with all SDF properties.
We state these constraints mathematically and simplify them significantly to make them computationally tractable.
We introduce an efficient algorithm that can interpolate a coarse discrete SDF while ensuring that every SDF value generated is consistent with both the input and all previously generated values.
This algorithm is further sped up by parallelized GPU preprocessing and parallelized CPU postprocessing.

We investigate several applications of our interpolation technique (see \autoref{fig:teaser}): global SDF refinement, SDF mesh reconstruction, and pseudo-SDF repair.
For refinement, \newchanged{given a uniform grid discrete SDF}{} we use a hierarchical data structure to \newchanged{interpolate additional, more finely spaced}{} SDF values on a narrow band around the zero level set.
By ensuring consistency with the entire SDF input, this effectively embeds global information into new, local, \newchanged{valid}{} SDF values.
For reconstruction, our refined SDF can simply be given as input to traditional isosurface extraction methods like Marching Cubes \cite{Lorensen1987,wyvill1986data}, yielding similar quality reconstructions to the state of the art \cite{rfta,kohlbrenner2025isosurface,kohlbrenner2025polyhedral},
but with added theoretical guarantees.
Lastly, with minimal changes, our method can be used to repair \emph{pseudo-SDFs} \cite{Marschner2023}, \newchanged{which consist of data satisfying the eikonal condition but not the distance property. Our method converts them into valid discrete SDFs,}{} enabling their use in downstream geometric processing applications. Our pseudo-SDF repair method is especially efficient because it is amenable to parallelization, and thus significantly improves over the state of the art.
\newchanged{With some caveats, these three applications can also be adapted to the more general setting of (discrete) conservative SDFs, where the data only provides distance bounds.\footnote{Since arbitrary conservative SDFs might satisfy neither the distance property nor eikonality, the resulting quality of repair, upsampling, and reconstruction depends on how much these properties are violated (see \autoref{fig:conservative-sdf}).}}{conservative-sdf}
We validate our algorithmic and parametric choices on a large set of 2D and 3D examples. 
By comparing our method to other conceivable alternatives for SDF refinement, reconstruction, and repair, we confirm experimentally our work's uniqueness in \emph{exactly} preserving all SDF properties
and ensuring that the modified SDF correctly corresponds to a consistent and feasible continuous surface.

\begin{figure}
\centering
\includegraphics[width=3.365in]
{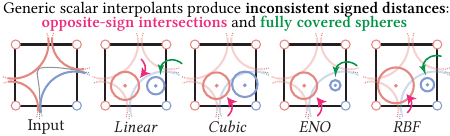}
\caption{Like linear interpolation, typical higher order polynomial interpolation strategies also regularly yield contradictory SDF values, with intersecting opposite-sign spheres or fully covered spheres.}
\label{fig:scalar-interpolants}
\end{figure}

\section{Related Work}
\label{sec:related}

To the best of our knowledge, no prior work considers the specific question of interpolating \newchanged{exact}{} signed distance data while precisely preserving its global consistency and mathematical properties.
To motivate our approach, we examine three possible ways one might repurpose existing methods to produce similar interpolated outputs, and show why each is ultimately ill-suited for this task.

\subsection{Scalar interpolation methods}
\label{subsec:interpolation}
Without considering the specifics of SDFs, interpolating general scalar field data from three-dimensional samples is a well-studied research question relevant in fields from numerical analysis to engineering, physics, and applied mathematics.
As shown in \autoref{fig:scalar-interpolants} \changed{(as well as by \citet{vadcubic}, for the specific case of cubic interpolants, and by \citet{baboud2011precomputed}, for bilinear interpolants)}{cite-baboud2011precomputed}, these general-purpose methods fail to preserve the mathematical properties of SDFs, and are therefore not directly applicable to our problem setting.
While a thorough review of this area is beyond the scope of this paper, common strategies include trilinear and tricubic interpolation from grids, the ENO \cite{harten1987uniformly} and weighted ENO \cite{liu1994weighted} schemes for numerical PDEs, and techniques like Radial Basis Functions (RBF) \cite{buhmann2000radial} and Moving Least Squares \cite{levin1998approximation} for unstructured data inputs.
These general-purpose polynomial schemes have been used to query \emph{approximate} SDFs from cached data (e.g., for collision detection \cite{koschier2017hp}), without guarantees of global consistency.

Beyond SDFs, a wide variety of scalar field interpolants have been proposed for a myriad of different graphics tasks.
For example, scalar fields are discretized using scattered particles and radial kernel interpolations in applications from implicit surface reconstruction \cite{kazhdan2005reconstruction,Kazhdan2006} to skinning \cite{dodik2024robust} to fluid simulation \cite{koschier2022survey}.
Significant work has also been dedicated to designing more accurate interpolation schemes for hybrid particle-grid simulations, such as the Material Point Method \cite{jiang2015affine,gao2017gimp}.
More spiritually related to our work are methods from the grid-based fluid simulation community in which grid velocities are interpolated in a manner that yields continuous pointwise divergence-free flows \cite{chang2021curl,schroeder2022local}.
In an analogous way, we set out to design an interpolant that satisfies the mathematical properties of SDFs; however, whereas the incompressible velocity interpolation of Schroeder et al.\ is strictly local, we will need to ensure \emph{global} consistency of the interpolated SDF values, introducing a new set of mathematical and computational challenges.

\begin{figure}
\centering
\includegraphics{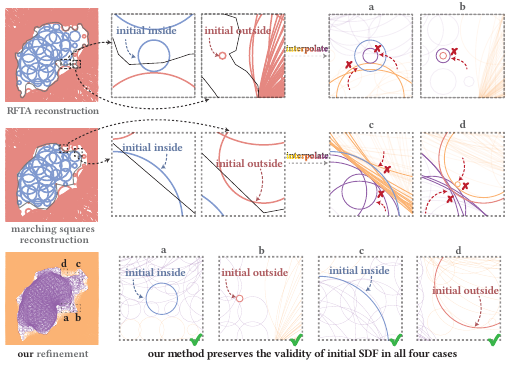}
\vspace{-0.7cm}
\caption{
Using
reconstruction methods for SDF \newchanged{refinement (purple and yellow spheres on the right)}{} by measuring the distance to the reconstructed surface on a refined grid \emph{(redistancing)} can violate the initial SDF: opposite-sign output spheres intersect \newchanged{and spheres fully contain others}{}.
Our method produces a refined SDF without reconstructing first and guarantees a valid SDF that is perfectly consistent with the input.}
\label{fig:redistancing}
\end{figure}

\begin{figure*}
    \centering
    \includegraphics{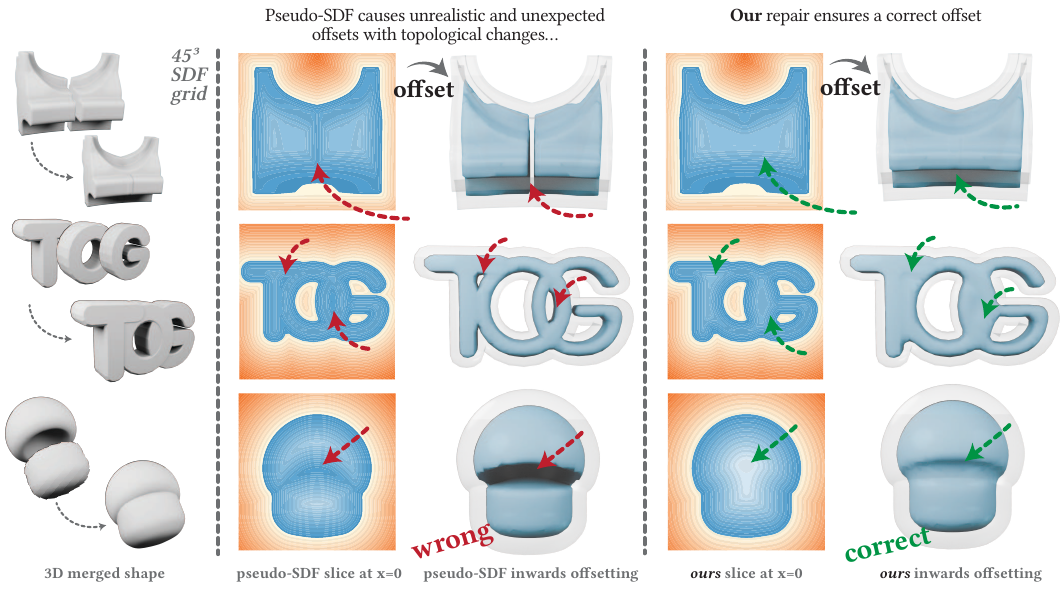}
    \caption{Computing the \(\omega\)-offset of a pseudo-SDF by meshing the \(\omega=0.15\) level set gives incorrect results. Only by repairing the pseudo-SDF with our method can we use the SDF to correctly compute offset surfaces}
    \label{fig:offset-3d}
\end{figure*}

\begin{figure*}
\centering
\includegraphics{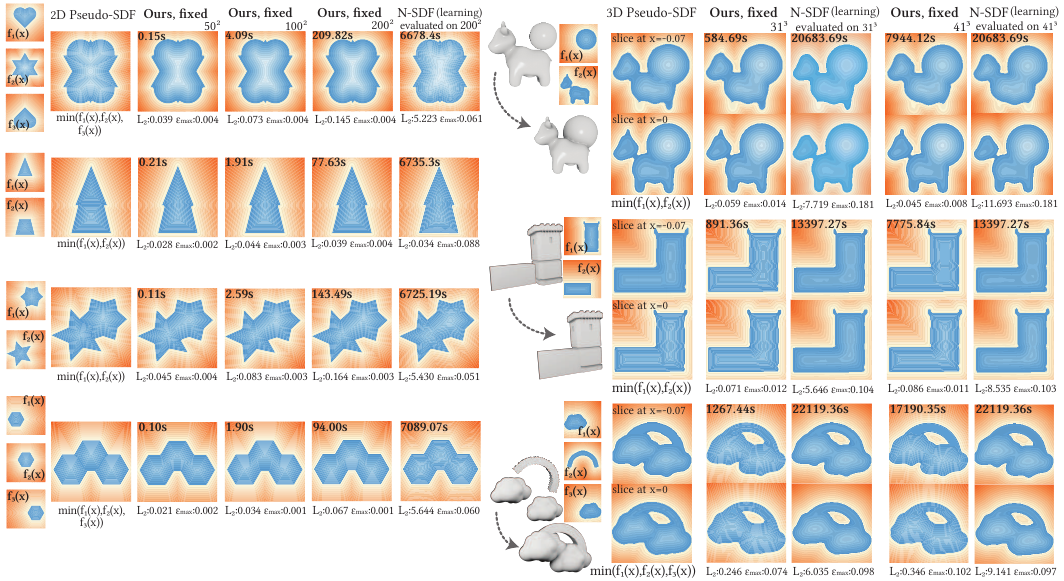}
\vspace{-0.7cm}
\caption{\changed{Given a pseudo-SDF composed of multiple simple shapes, our method can repair it at different resolutions while using significantly less time than the neural method of \citet{Marschner2023} with lower maximum error \(\epsilon_{max}\) and \(L_2\) error.
Runtime for \citet{Marschner2023} is the same for all resolutions, as we evaluate the same model (trained with default parameters) on different grids.}{new-experiments-added}}
\label{fig:marschner}
\end{figure*}

\subsection{Reconstruction and redistancing}
\label{subsec:reconstruction-redistancing}
\newchanged{Extracting explicit meshes from a discrete set of SDF values has been an active area of research for decades.}{}
Thus, a plausible interpolation procedure is to reconstruct an explicit surface from the input samples and directly compute the signed distance to that surface.
We call this strategy \emph{redistancing}.
Unfortunately, as noted recently by \citet{rfts}, classical SDF reconstruction methods like Marching Cubes \cite{Lorensen1987} and Dual Contouring \cite{Ju2002} produce reconstructions without accounting for the global information provided by the SDF samples, leading to contradictions between input and redistanced data (see \autoref{fig:redistancing}, middle).
Indeed, even the work of \citet{rfts}, which identifies the three properties that a surface must satisfy to be globally consistent with an SDF, only imposes these as soft constraints, and can lead to the same type of contradictions;
the same holds for the work of \citet{rfta} (see \autoref{fig:redistancing}, top), as well as \newchanged{recent methods}{} by \citet{kohlbrenner2025isosurface,kohlbrenner2025polyhedral}, which guarantee two of the constraints at the expense of the third.
\newchanged{More recent methods  \cite{kohlbrenner2026gradient,cheng2026rbfsdfradialbasis} exploit these geometric constraints in locating surface points and thereby offer improved reconstructions; however, neither method guarantees that the resulting surfaces satisfy all three global consistency constraints.}{new-paper-added}

\newchanged{Redistancing techniques in the level set method literature compute approximate SDF values by marching the distance field outwards from the zero level set \cite{sethian1996fast} or numerically solving the eikonal equation \cite{sussman1999efficient,cheng2008redistancing}, without necessarily using an explicit surface; related Euclidean Distance Transform (EDT) methods can compute exact distances only when zero-valued seeds coincide with grid points \cite{felzenszwalb2012distance}. Other work constructs conservative distance bounds for shape approximation and efficient implicit-model evaluation \cite{sharp2022spelunking,schott2025sphere,barbier2025lipschitz,hubert2025accelerating}. Closer to our work, \citet{baboud2011precomputed} precompute ``free-space'' regions for efficient heightfield rendering. However, all methods mentioned above either never consider consistency or apply traditional interpolation methods that lead to inaccuracies. Unlike these approaches, our goal is to interpolate SDF samples while guaranteeing consistency with the given discrete data.}{paragraph-rewritten} 

\subsection{Neural Signed Distance Functions}
\label{subsec:neuralsdf}
Recently, SDFs have emerged as one of the representations of choice in deep learning applications like shape compression \cite{davies2020effectiveness}, completion \cite{Park2019}, real-time rendering \cite{takikawa2021nglod} and even geometric reconstruction from images \cite{wang2021neus} and cross-sections \cite{walker2025crosssdf}.
More relevant to our research question is a significant body of work that has proposed using neural networks as parametric function spaces. Gradient-based optimization is performed on the network's parameters to best fit some observed data (e.g., point clouds or occupancy data) while minimizing some differential regularization losses \cite{xi2025neuralssd,coiffier20241,bethune2023robust,wang2023neural,ben2022digs,lipman2021phase,sitzmann2020implicit,atzmon2020sal,atzmon2020sald,fayolle2021signed}.
Recent work also explores reconstructing surfaces directly from neural SDFs:
\citet{liu2025direct} perform feature-preserving isocontouring using a Voronoi-based approach that assumes access to a continuous SDF, and \citet{stippel2025marching} analytically extract an explicit mesh by exploiting the internal architecture of a given neural SDF.
One could imagine repurposing these methods to be supervised on observed SDF data, employing sampling of the trained network as an interpolation strategy, albeit a computationally expensive one prone to local minima. \newchanged{Similarly, learning-based methods have been proposed for SDF upsampling, including the diffusion-based SDF super-resolution method by \citet{shim2023diffusion} for generative shape modeling, and SuperSDF \cite{supersdf}, which constructs high-resolution SDF grids from coarse SDF inputs using sparse voxel neural networks.}{new-paper-added}
Unfortunately, due to the use of soft constraints in the form of network losses (e.g., the heat loss of \citet{wang2025hotspot}), these methods cannot guarantee adherence to the input and the SDF properties.
Even when some of these constraints may be enforced constructively via the network architecture (e.g., the Lipschitz network designed by \citet{coiffier20241}), these methods consider only local (e.g., differential) properties and not the global consistency of the output SDF.

An exception is the recent work by \citet{Marschner2023}, which proposes a global loss that enforces the mathematical structure of Signed Distance Fields. This loss is used to repair \emph{pseudo-SDFs}: functions that locally \emph{look} like SDFs (i.e., satisfy eikonality, \newchanged{meaning that the norm of the gradient is unit everywhere}{}) but are not actual distance fields, which are often produced by Boolean operations
\changed{and CSG trees \cite{reiner2011interactive}.}{csg}
Indeed, this type of repair operation is one of many possible applications of our more general interpolation framework \changed{(as \autoref{fig:offset-3d} shows).
\autoref{fig:marschner} shows that our method can function}{moved-offset-up}
as a non-neural replacement for the pseudo-SDF repair of \citet{Marschner2023} for discretely sampled pseudo-SDFs, avoiding the large training times and memory requirements of neural networks, and strictly enforcing global consistency without the need to soften constraints.

\newchanged{\section{Overview}}{}
\begin{newchangedblock}
To set the stage for the presentation of our method's theory, implementation, and  applications to follow, we begin by summarizing our overall approach.

\begin{itemize}
    \item Viewing each SDF sample as a signed sphere, we first define the notion of \emph{discrete SDF validity}, which determines whether there can exist any surface giving rise to that discrete SDF (\autoref{sec:definitions}).

    \item We then introduce a validity-preserving SDF interpolation strategy in \autoref{sec:single-point} and \autoref{sec:scoring}. At a new query point, we construct a finite set of candidate SDF values based on the relationships among spheres, discard inconsistent candidates that violate discrete SDF validity, and use a scoring scheme to select among the remaining valid values.

    \item We apply our interpolation strategy to \emph{SDF refinement} in \autoref{sec:application-refinement}. We propose a Dual Octree SDF (DOS) data structure to adaptively subdivide cells near the surface, and demonstrate several refinement orderings. We further introduce sphere culling (which can violate some of our guarantees, but is needed for large examples), GPU preprocessing (without which finding intersections between spheres would be prohibitively expensive), and other performance improvements to make the method practical.

    \item We apply the SDF refinement result to \emph{surface reconstruction} in \autoref{sec:application-reconstruction}. The finest level of the DOS provides a minimal, high-resolution narrow band of globally consistent and locally sampled SDF values. We augment this band with sufficient additional high-resolution uniformly spaced samples to apply Marching Cubes and efficiently extract the surface.

    \item Finally, we address \emph{pseudo-SDF repair} in \autoref{sec:application-pseudosdf}. We identify invalid SDF samples whose spheres are fully covered, and repair them by assigning each the \emph{minimum valid radius} consistent with the remaining valid samples.
\end{itemize}
\end{newchangedblock}

\section{\newchanged{SDF Interpolation: Theory and}{} Method}
A Signed Distance Function (SDF) of \changed{a surface \(\Omega\) that encloses a volume $\Sigma$ (i.e., $\Omega = \partial \Sigma$) }{interior-1} is a function \(\sdf : \mathbb{R}^d \rightarrow \mathbb{R}\) (\(d=2,3\)) such that
\changed{\begin{equation}\label{eq:smoothsdf}
    \sdf(\smoothpt) = \begin{cases}
        - \distance (\smoothpt, \surface) & \text{if }\quad x \in \Sigma \,, \\
        \distance (\smoothpt, \surface) & \text{otherwise}\, .
    \end{cases}
\end{equation}}{interior-2}

The input to our method has two components:
\begin{enumerate}
    \item SDF samples \(\sdfsample_1, \dots, \sdfsample_n \in \mathbb{R}\) at points \(\point_1, \dots, \point_n \in \mathbb{R}^d\) corresponding to an SDF \(\sdf\), i.e., \(\sdf(\point_i) = \sdfsample_i \; \forall i\) (a \emph{discrete SDF});
    \item a new, different set of query points $\point_{n+1},\dots,\point_{m}$ without known $s_i$ values.
\end{enumerate}
We write the discrete SDF as a collection of pairs,
\(
    \big\{ (\point_i,\sdfsample_i) \big\}_{i=1}^{n}
    \;.
\)
Our ultimate goal is to characterize the set of possible signed distance values $\sdfsample_{n+1},\dots,\sdfsample_{m}$ for query points $\point_{n+1},\dots,\point_{m}$ that are consistent with the inputs $\big\{(\point_i,\sdfsample_i)\big\}_{i=1}^{n}$ and with each other, and choose out of all these possibilities the values that best match some prior knowledge or heuristic about $\surface$. \changed{Critically, we will produce these new signed distance values without assuming any access to the original SDF $\sdf$ beyond the original samples $\big\{(\point_i,\sdfsample_i)\big\}_{i=1}^{n}$.}{no-access-to-sdf-2} 

\begin{figure}
\centering
\includegraphics{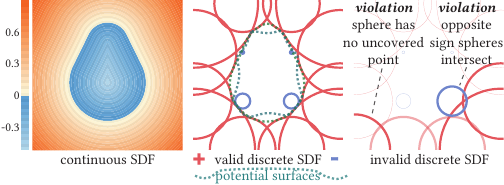}
\vspace{-0.8cm}
\caption{A continuous SDF \emph{(left)};
a discrete SDF visualized as spheres with multiple valid fitting surfaces \emph{(center)};
and an invalid SDF \emph{(right)}.}
\label{fig:sdf-didactic}
\end{figure}

We will produce these values sequentially, ensuring that any newly assigned signed distance $\sdfsample_{n+j}$ remains consistent with both the input $\{(\point_i,\sdfsample_i)\}_{i=1}^{n}$ and all the values $\{(\point_i,\sdfsample_i)\}_{i=n+1}^{n+j-1}$ generated previously.
Thus, it will suffice to answer the simpler question of finding the best possible signed distance value for a single new point $\point$.
We will answer this question in two steps: first, we will build the mathematical infrastructure necessary to characterize all possible signed distance values at $\point$ (Sections \ref{sec:definitions} and \ref{sec:single-point}); then, we will introduce scoring rules that enable our algorithm to select a desirable distance value out of all possibilities (\autoref{sec:scoring}). 

\subsection{Valid and invalid SDF values}
\label{sec:definitions}

Every signed distance sample \((\point_i,\sdfsample_i)\) of a surface $\Omega$ defines a sphere centered at \(\point_i\) with radius \(|\sdfsample_i|\) and sign $\sign(\sdfsample_i)$.
\changed{This radius, of course, is nonnegative;
we, nevertheless, sometimes casually use the language ``sign of the sphere'' to refer to the sign of \(\sdfsample_i\).}{sign-of-sphere-clarified}
As shown by \citet{rfts}, by definition of signed distance, given such a sample $\Omega$ must satisfy the following (see Fig.~\ref{fig:sdf-didactic}, \emph{center}):

\begin{enumerate}
    \item \(\surface\) does not penetrate the sphere (i.e., it does not intersect the corresponding open ball, so $\Omega\cap \mathring{B}(\point_i,|\sdfsample_i|) = \emptyset$),
    \item \(\surface\) touches the sphere (i.e., $\Omega\cap \partial B(\point_i,|\sdfsample_i|) \neq \emptyset$), and
    \item \(\surface\) contains the sphere if it has negative sign, and does not contain it otherwise (i.e., $\partial B(\point_i,|\sdfsample_i|) \ \subseteq \Omega \iff \sdfsample_i<0$),
\end{enumerate}
where $B(\mathbf{a},b)$ is the closed ball centered at $\mathbf{a}$ with radius $b$, \newchanged{$\mathring{B}$ is the interior of the ball, and $\partial B$ is the surface of the ball}{new-definition}.
Any surface $\surface$ that satisfies all conditions (1-3) for all spheres determined by a discrete set of signed distances $\big\{(\point_i,\sdfsample_i)\big\}_{i=1}^{n}$ is a valid reconstruction of the discrete SDF.
Conversely, we will say that an arbitrary set $\big\{(\point_i,\sdfsample_i)\big\}_{i=1}^{n}$ is \emph{valid} if there exists an $\surface$ that satisfies the conditions (1)-(3) for all of its corresponding spheres. The set is \emph{invalid} if it contains contradictions that make the existence of such an $\surface$ impossible.
Our first theoretical result is the observation that the validity of a discrete SDF can be completely characterized geometrically in terms of spheres and their intersections, without building or testing any surfaces.

\begin{proposition}[Discrete SDF validity]
\label{prop:discrete-sdf-validity}
A discrete SDF is valid iff (Fig.~\ref{fig:sdf-didactic}, \emph{right}):
\begin{enumerate}
    \item[(i)] No differently signed spheres intersect (they are allowed to touch at a single point).
    \item[(ii)] Every sphere has at least one point that is not in the ball enclosed by any other sphere (an \emph{uncovered point}).
\end{enumerate}
\end{proposition}
\begin{proof}
\changed{%
    Let \( \{ (\point_i,\sdfsample_i)\} _{i=1}^n\) be a valid SDF.
    Then, by definition, there is a surface \(\Omega\) such that \( \{ (\point_i,\sdfsample_i)\} _{i=1}^n\) is its SDF.
    Since different-signed spheres are on opposite sides of \(\Omega\), and their distance to \(\Omega\) is exactly their radii, they can at most touch each other, and never intersect, and thus \( \{ (\point_i,\sdfsample_i)\} _{i=1}^n\) fulfills the first condition.
    Since \(\Omega\) is tangent to every sphere, but does not penetrate any sphere, every sphere has at least one uncovered point (the tangent point), and thus \( \{ (\point_i,\sdfsample_i)\} _{i=1}^n\) fulfills the second condition.
}{discrete-sdf-validity-proof-1}

\changed{
    Let \( \{ (\point_i,\sdfsample_i)\} _{i=1}^n\) be a set of real values and points fulfilling the conditions of the proposition.
    Let \(C_-\) be the contour of all inside (negative) spheres, and let \(C_+\) be the contour of all outside (positive) spheres.
    By the first condition, these contours never intersect, but at most touch.
    Hence, there is a power contour, a polyhedron formed by the facets of the power diagram (formed by the power distance \(\Pi(i,j) = \lVert \point_i - \point_j \rVert^2 - |\sdfsample_i| - |\sdfsample_j|\)) that separates the disjoint inside and outside contours \cite[Section 4.5]{kohlbrenner2025isosurface}, and that never penetrates either contour.
    By the second condition, every sphere \((\point_i,\sdfsample_i)\) contributes at least one point \(\otherpoint_i\) to either \(C_-\) or \(C_+\).
    Let \(\tilde{\Omega}\) be the surface formed by the power contour (oriented so that \(C_-\) is inside and \(C_+\) is outside \(\tilde{\Omega}\)).
    For each \(\otherpoint_i\) we add a slender spike connecting a power contour face corresponding to sphere \(i\) with \(\otherpoint_i\) to \(\tilde{\Omega}\), creating the surface \(\Omega\).
    Each such spike must remain entirely within the uncovered region of that power cell until it meets the point \(\otherpoint_i\), which is possible for at least one power contour face per \(i\).
    This resulting surface \(\Omega\) is tangent to every sphere (because it passes through all \(\otherpoint_i\)), it does not intersect any sphere, all the negative spheres are inside, and all the positive spheres are outside.
    Hence \( \{ (\point_i,\sdfsample_i)\} _{i=1}^n \) is a valid discrete SDF.
}{discrete-sdf-validity-proof-2}
\end{proof}

\subsection{Finding \emph{a} valid SDF value for \emph{a single} new point \(\point\)}
\label{sec:single-point}

Let us now consider the simplified single-point interpolation case, where the set $\{(\point_i,\sdfsample_i)\}_{i=1}^{n}$ is valid and one wishes to find a signed distance value (signed radius) $\sdfsample$ for a new position (sphere center) $\point$ that makes the augmented set $\big\{(\point_1,\sdfsample_1), \dots, (\point_n, \sdfsample_n), (\point, \sdfsample) \big\}$ valid.
Two key challenges appear when attempting to use \autoref{prop:discrete-sdf-validity} to find such a value $\sdfsample$.
First, condition (ii) is quite cumbersome to check in practice: in 3D, a naive implementation requires finding all uncovered intersections of triples of spheres, which is \(O(n^4)\) in general.
More critically, while \autoref{prop:discrete-sdf-validity} lets us verify whether or not a configuration is valid, it does not provide instructions to \emph{select} a valid radius $\sdfsample_{n+1}$ out of the infinite set of possible choices.
In this section, we deduce one such constructive strategy.

\begin{figure}
\centering
\includegraphics{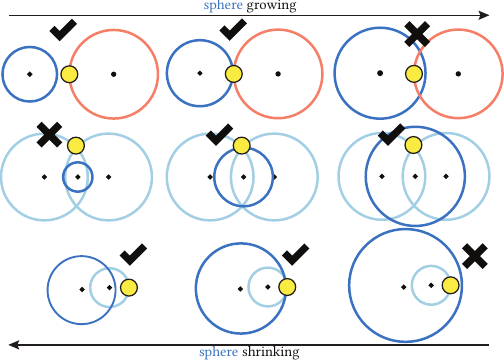}
\vspace{-0.7cm}
\caption{\changed{%
Adding a \textcolor{solidblue}{solid}-\textcolor{solidblue}{colored} sphere in the presence of existing \textcolor{lightblue}{lig}\textcolor{lightred}{ht}-\textcolor{lightblue}{col}\textcolor{lightred}{ored} spheres.
The growing or shrinking sphere in \(d=2\) can invalidate an SDF by (1) introducing an opposite-sign sphere intersection at a tangent point \emph{(top)}, (2) losing its only remaining uncovered point at a tangent point or an existing uncovered intersection \emph{(middle)}, and (3) fully covering a different sphere at either a tangent point or an existing uncovered intersection \emph{(bottom)}.}{better-explaining-didactic-growing}}
\label{fig:growing-spheres}
\end{figure}

Conceptually, we start with a sphere of infinitesimal radius centered at the query point $\point$, and grow the \changed{radius of the sphere}{saxon-1} $|s|$ until it contacts a \emph{grow-to point} \(\otherpoint \in \mathbb{R}^d\), i.e., a point at which SDF validity \emph{may} change.
At \(\otherpoint\), we check whether the conditions of Prop.~\ref{prop:discrete-sdf-validity} hold for the revised set \(\big\{(\point_1,\sdfsample_1), \dots, (\point_n, \sdfsample_n), (\point, \sdfsample) \big\}\). If they do, then \(\sdfsample\) is a valid SDF value, otherwise, we can continue on to check the next grow-to point.
It turns out that there are finitely many grow-to points we need to check in order to find a valid SDF value \(\sdfsample\).
The intuition is the following:
for a growing or shrinking sphere to change the discrete SDF's validity, the sphere has to 
(see \autoref{fig:growing-spheres})
\begin{itemize}
    \item gain or lose an intersection with an existing opposite-sign sphere, which can only happen if it grows or shrinks past a potential uncovered tangent point that it shares with the opposite-sign sphere; 
    \item gain or lose its only uncovered point, which can only happen if it grows or shrinks past a potential uncovered tangent point it shares with an existing sphere, or if it grows or shrinks past an uncovered intersection point of multiple existing spheres; or
    \item create or destroy the only uncovered point of an existing sphere, which similarly can only happen if it grows or shrinks past tangent or intersection points.
\end{itemize}

The following definition summarizes all the points at which the validity of the SDF can change (in addition to the trivial radius \(0\)).

\begin{wrapfigure}[18]{r}{0.22\columnwidth}
    \vspace{-\intextsep}
    \hspace*{-20pt}
    \includegraphics[width=1.3\linewidth]{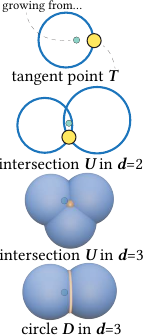}
\end{wrapfigure}

\begin{definition}[Candidate grow-to points]
\label{def:grow-to-points}
    \quad

    Let \(T = \left\{\point_i \pm \frac{\point - \point_i}{\lVert \point -\point_i \rVert} |\sdfsample_i| \;\big|\; \forall i = 1, \dots, n\right\}\) be all points that our sphere $(\point,s)$ could grow to that would make it tangent to an existing sphere.

    Let \(U\) be all uncovered intersection points on at least two spheres if \(d=2\), or on at least three spheres if \(d=3\) \cite{Fang1986}.

    If \(d=3\), we also need to consider circles at which two spheres intersect.
    We define  \(\mathbf{v}_{ij}=\point - \boldsymbol{\zeta}_{ij}, \mathbf{w}_{ij}=(\mathbf{v}_{ij} \cdot \boldsymbol{\mu}_{ij})\boldsymbol{\mu}_{ij}\)
    where
    \(\boldsymbol{\zeta}_{ij},\rho_{ij},\boldsymbol{\mu}_{ij}\) are, respectively, the center, radius and normal of the intersection circle of two spheres \(i,j\) if they intersect in a circle that is not completely covered.
    Then, let \(D = \left\{ \boldsymbol{\zeta}_{ij} \pm \frac{\mathbf{v}_{ij}-\mathbf{w}_{ij}}{\lVert \mathbf{v}_{ij}-\mathbf{w}_{ij} \rVert} \rho_{i,j} \;\big|\; \forall i \neq j, \; i,j=1,\dots,n \right\}\)
    be the points that our sphere could grow to that would make it tangent to the uncovered circle formed by the intersection of two existing spheres. 

    Finally, \(P = T \cup U \cup D \cup \{p\}\) is the finite set of all candidate \emph{grow-to points}: the point itself and all points at which a growing sphere has the potential to change the validity of the SDF.
\end{definition}

The relevance of $P$ becomes obvious with our next theoretical result, which guarantees the existence of a point our new sphere can grow to that preserves the validity of the input configuration. In other words, the set of \emph{valid grow-to points}, which we denote $\hat{P}$, is always non-empty. This turns the continuous search for a valid distance value $\sdfsample$ into one over the finite set $P$.

\begin{proposition}[Grow-to point validity]\label{prop:validradii}
    There exists at least one \(q \in P\) such that \(\sdfsample=\lVert q-\point \rVert\) or \(-\lVert q-\point \rVert\) is a valid radius at \(\point \in \mathbb{R}^d\), i.e., \(\big\{(\point_1,\sdfsample_1), \dots, (\point_n, \sdfsample_n), (\point, \sdfsample) \big\}\) is a valid SDF.
    We call the set of such valid grow-to points $\hat{P}\subseteq P$.
\end{proposition}
\begin{proof}
\changed{%
    If \(\point\) is not contained within any existing sphere,
    then it can have SDF value \(s=0\) and be a valid SDF sample.
    The closest possible \(\otherpoint \in T\) to \(\point\) also results in an SDF sphere \((\point,\sdfsample)\) that at most touches any existing sphere, and thus does not intersect any opposite-sign sphere, has an uncovered point, and does not cover any existing spheres' points.
}{grow-to-point-validity-proof-1}

\changed{%
    Assume therefore that \(\point\) is contained within an existing sphere (without loss of generality, a positive sphere).
    The distance \(\sdfsample\) between \(\point\) and the contour of all existing positive spheres \(C_+\) is a valid SDF value at \(\point\), since
    \begin{itemize}
        \item the negative spheres' contour \(C_-\) and the positive spheres' contour \(C_+\) can at most touch, i.e., \((\point,\sdfsample)\) cannot intersect any negative sphere;
        \item \((\point,s)\) has an uncovered point, namely the point at which it touches \(C_+\); and
        \item since the addition of \((\point,\sdfsample)\) does not change the positive spheres' contour, no new points of existing spheres are covered.
    \end{itemize}
}{grow-to-point-validity-proof-2}

\changed{%
    There is a grow-to point \(\otherpoint \in P\) that corresponds to the sphere radius \(\sdfsample\), namely it is the closest point to \(\point\) on the contour \(C_+\).
    Since \(C_+\) consists of sphere arcs (\(d=2\)) or patches (\(d=3\)), the closest point to \(\point\) is either the point \(\point_i \pm \frac{\point-\point_i}{\lVert \point-\point_i \rVert} |\sdfsample_i|\) in the middle of the arc/patch of sphere \(i\) (one of the points in \(T\)), or it is at a boundary of an arc/patch.
    For \(d=2\), the boundaries of all arcs are in \(U\).
    For \(d=3\), the boundaries of all patches are arcs themselves;
    the closest points to \(\point\) on these arcs are either in the middle of these arcs (\(D\)) or the boundaries of these arcs (\(U\)).
}{grow-to-point-validity-proof-3}
\end{proof}

\newchanged{Notably, interpolating with any valid radius satisfying \autoref{prop:validradii} is a \emph{true} interpolation:
In the limit of \(\point \to \point_i\), all valid radii will approach \(\sdfsample\). A proof of this statement can be found in \autoref{app:additionalproofs}}{new-proofs}.
\autoref{prop:validradii}'s constructive proof is the basis for our strategy to compute the \emph{closest valid grow-to point} $\qmin \in \hat{\growtopoints}$. Even further, it proves that $\rmin = \|\point_{n+1} - \qmin \|$ is exactly the \emph{minimum valid distance value} at $\point_{n+1}$.
In Sec.~\ref{sec:application-pseudosdf}, we will exploit a critical property of this minimum value: namely, that it can be computed independently for a large number of new points. We will use this property to design a simple, parallel algorithm to repair faulty signed distance data. 
While this minimum radius value can be computed very efficiently, naively assigning it to every new point $\point$ is far from the best interpolation strategy in most applications.
We thus explore the set $\hat{P}$ of valid grow-to points, and select from it a more desirable element.

\begin{figure}
\centering
\includegraphics{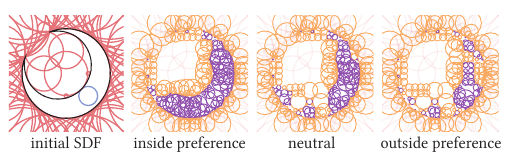}
\vspace{-0.9cm}
\caption{We score radii during refinement and reconstruction so that inside spheres are slightly preferred \emph{(center left)}, because we find that it often provides qualitatively better results than no preference \emph{(center right)} or preferring outside spheres \emph{(far right)}.}
\label{fig:radius-scores}
\end{figure}

\subsection{Finding a \emph{good} valid SDF value for a single new point \(\point\)}
\label{sec:scoring}

In practice, the set of possible signed distance values at $\point$ given by the valid grow-to points $\hat{P}$ is not only non-empty (\autoref{prop:validradii}), but large:
different choices of grow-to points will discriminate between all the possible surfaces that are consistent with the input SDF.
In this section, we explore this choice as an opportunity to encode any prior knowledge we may have about the true, unknown surface $\surface$.

\begin{wrapfigure}[10]{r}{0.2\columnwidth}
    \vspace*{-2pt}
    \hspace*{-20pt}
    \centering
    \includegraphics[width=1.2\linewidth]{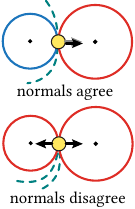}
\end{wrapfigure}
\noindent
We decide which valid point to pick by assigning a
\emph{score} to each $\otherpoint\in\hat{P}$ that can be used to
incorporate a variety of application-specific priors.
For general-purpose refinement and reconstruction, we stick to a prior
popular in previous work: \emph{smoothness}.
In smooth surfaces, the normal varies continuously across the surface.
If the surface $\surface$ is tangent to the sphere
$(\point,\sdfsample)$ at a specific grow-to-point $q$, its normal
direction at $\otherpoint$ is the direction of the vector
$\point-\otherpoint$ (the reverse if the
\changed{sign of the sphere}{saxon-2} is negative).
When a growing sphere hits a grow-to-point $\otherpoint$, we want to
score this grow-to-point highest if it results in the smallest normal
variation, and thus a smoother surface.

\begin{figure}
\centering
\includegraphics{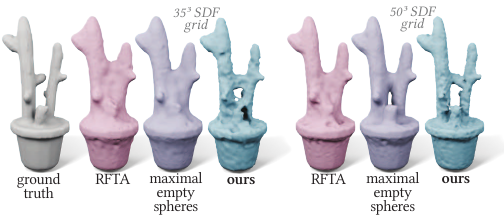}
\vspace{-0.8cm}
\caption{\changed{Methods like RFTA \cite{rfta} tend to oversmooth reconstructed surfaces.
Our method is better at recovering sharp edges and detailed features.}{prior-language-in-figure}}
\label{fig:rfta-prior}
\end{figure}

Combined, the previous three sections propose a strategy to interpolate a signed distance value for a new point $\point$ in a way that is consistent with a set of input samples $\{(\point_i,\sdfsample_i)\}_{i=1}^{n}$. A pseudocode for our interpolation algorithm is provided below.
In the sections that follow, we review several different applications of this strategy (Sections \ref{sec:application-reconstruction}, \ref{sec:application-refinement} and \ref{sec:application-pseudosdf}), and discuss performance improvements that eliminate the most computationally intensive steps (\autoref{sec:performanceimprovements}).

\begin{lstlisting}
//This function interpolates an SDF to the point p.
function interpolate_sdf_to(p, p1, ..., pn, s1, ..., sn):
  radii = []
  for (q,pi,si) in grow_to_points(p,p1,...,pn,s1,...,sn):
    for sign=-1,1:
      s = sign * len(q-p)
      radii.add((score_for_radius(p,s,pi,si,q), s))
  sort(radii) //in descending order
  for (score,s) in radii:
    if is_sdf_valid(p1,...,pn,p,s1,...,sn,s):
      return s
\end{lstlisting}

\changed{%
The pseudocode for computing a list of grow-to points is:
}{pseudocode-moved-2}
\begin{lstlisting}
//For an input interpolation point p and an SDF, this function returns a list of triples (q,pi,si) of a candidate grow-to point and the existing sphere the candidate is on.
function grow_to_points(p, p1, ..., pn, s1, ..., sn):
  points = []
  for i = 1,...,n:
    //Add tangent points
    for c = -1,1:
      q = pi + c*normalized(p-pi)*|si|
      points.add((q,pi,si))
    //Add uncovered intersection points
    for j = i+1,...,n:
      if d==2:
        if q1,q2 = sphere_intersect_2d(pi,|si|,pj,|sj|):
          for q = q1,q2:
            if point_uncovered(q,p1,...,pn,s1,...,sn):
              points.add((q,pi,si), (q,pj,sj))
      else if d==3:
        for k = j+1,...,n:
          if q1,q2 = sphere_intersect_3d(pi,|si|,pj,|sj|, pk,|sk|):
            for q = q1,q2:
              if point_uncovered(q,p1,...,pn,s1,...,sn):
                points.add((q,pi,si), (q,pj,sj), (q,pk,sk))
        //In 3D, add circle points
        if zeta,rho,nu = intersection_circle_uncovered(pi, |si|,pj,|sj|):
          q1 = project_to_closest_pt_on_circle(p,zeta,rho,nu)
          q2 = project_to_farthest_pt_on_circle(p,zeta,rho,nu)
          for q = q1,q2:
            if point_uncovered(q,p1,...,pn,s1,...,sn):
 j              points.add((q,pi,si), (q,pj,sj))
  return points
\end{lstlisting}

\begin{figure}
\centering
\includegraphics{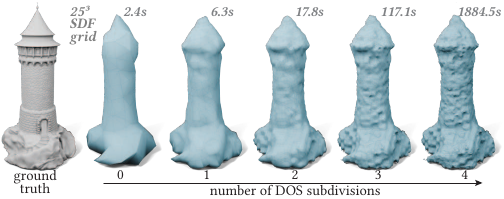}
\vspace{-0.7cm}
\caption{Generally, adaptively subdividing the DOS results in reconstructions with more detail, but with quickly diminishing returns \newchanged{and rapidly increasing computational cost}{}.
This informs our choice of default parameter for the subdivision depth \(\subdivisiondepth\) to be \(2\) or \(3\).}
\label{fig:subdivision-levels}
\end{figure}

\section{Application: Refining signed distance data}
\label{sec:application-refinement}

The previous sections analyzed the set of valid signed distance values $\sdfsample$ at a single new point $\point$ that are consistent with a set of valid input samples $\{(\point_i,\sdfsample_i)\}_{i=1}^{n}$.
In practice, however, one often wants to find valid signed distances $\sdfsample_{n+1},\dots,\sdfsample_{m}$ at a large set of new points $\point_{n+1},\dots,\point_{m}$ near $\surface$'s zero level set.
We refer to this as \emph{refinement}.
Importantly, our refinement task differs from the one considered, for example, by \citet{kohlbrenner2025isosurface}: we do not have access to the ground-truth continuous SDF function, and must instead interpolate.
We advocate for a \emph{greedy} algorithm that computes a new value $\sdfsample_{n+j}$ by ensuring its validity with respect to the input as well as to all previously generated values.
However, a naive repetition of the algorithm proposed in \autoref{sec:scoring} would be very expensive: in 3D, its complexity would scale with $\sim m \times n^{7}$ (accounting for each validity check of each grow-to point).

\subsection{The Dual Octree SDF (DOS)}
\label{sec:DOS}
\begin{wrapfigure}[12]{r}{0.3\columnwidth}
    \vspace*{-15pt}
    \hspace*{-20pt}
    \includegraphics[width=1.3\linewidth]{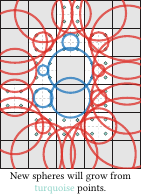}
\end{wrapfigure}
To avoid the naive refinement complexity, we will focus on the case in which the input samples lie on the cell centers of a regular grid.
We can immediately identify the set of \emph{interesting} grid cells that may contain the surface $\Omega$: those in which the center's stored distance value $\sdfsample$ measures less than half the cell's diagonal $\diagonal$ (otherwise the cell is fully covered by the sphere, and \(\Omega\) cannot be inside it).
Instead of \newchanged{upsampling}{} the entire regular grid, we will adaptively subdivide only these interesting cells in a data structure we call a \emph{Dual Octree SDF (DOS)} (see inset, where non-interesting cells are gray). 
We repeat this subdivision process for the whole grid \(\subdivisiondepth\) times, with \(\subdivisiondepth = 2\) for \(d=3, n>20^3\), and \(\subdivisiondepth = 3\) otherwise \newchanged{to balance the improved reconstruction accuracy from additional samples against the associated increase in computational cost}{} (see \autoref{fig:subdivision-levels}).

To maintain the validity of distance samples at the neighboring cells, the radius of a new sphere generated from a new point inside the interesting cell cannot be arbitrarily large.
\newchanged{We can put this fact to use: For a given interesting cell chosen for subdivision, we determine a bounding sphere centered on that cell that is guaranteed to contain any and all possible valid SDF spheres that could be assigned to its child cells. The radius of this bounding sphere, which we call the cell's \emph{coverage range}, is given by \(\max\left(|\sdfsample|+\diagonal/2, 3\diagonal/4\right)\), where \(\diagonal\) is the diagonal of the \newchanged{initial}{} cell, and \(|\sdfsample|\) is the radius of the cell's SDF sphere. We prove this bound in \autoref{app:additionalproofs}, \autoref{prop:maximal_coverage_range}.}{equation-updated}
Critically, this means that any \newchanged{existing}{} grid distance value \newchanged{whose SDF sphere does not enter this range}{} is 
\emph{irrelevant} to that cell \newchanged{(see \autoref{fig:max-child-radius}, the dashed red range)}{}.  
To exploit this, we initially gather all \emph{relevant} spheres for each \newchanged{interesting cell and, when interpolating new values for its child cells, we consider only these relevant spheres}{}, greatly reducing the number of spheres we must iterate over in our algorithm.
For all input spheres, we can also precompute (and dynamically update) the set of uncovered intersections and circles $U \cup D$, significantly reducing the asymptotic cost of both gathering candidate grow-to points and checking their validity.

\begin{figure}
\centering
\includegraphics{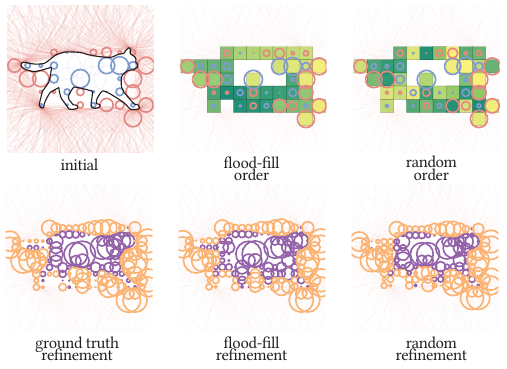}
\vspace{-0.7cm}
\caption{Our method performs subdivision on cells using a flood-fill ordering. Here we compare it to a random ordering, with our method performing slightly better in refining the tail and back legs.
The order of cells to be subdivided is indicated in shades of green; dark cells are subdivided first.}
\label{fig:subdivision-orders}
\end{figure}

\subsection{Refinement order}
\label{sec:flood-fill-refinement}

In our sequential refinement, the set of possible values at a new point $\point_{n+j}$ will depend not only on the input $\{(\point_i,\sdfsample_i)\}_{i=1}^{n}$ but also on the previously generated $\sdfsample_{n+1},\dots,\sdfsample_{n+j}$.
Importantly, this \emph{greedy} nature of our algorithm means that it will be sensitive to the ordering of the new sample points $\point_{n+1},\dots,\point_m$.
In general, the generated spheres will be larger in cells that are less covered by other spheres. These larger spheres often correspond to surfaces with simpler topologies. Additionally, the growth of a new sphere in a given cell will be most impacted by the spheres in the cells directly neighboring it.
These empirical observations motivate the following refinement ordering strategy:
we begin by estimating each interesting cell's \emph{covered area ratio} by sampling a small number ($8^d$) of regularly spaced points in the cell and querying their containment in the cell's relevant spheres.
Then, we identify the interesting cell with the lowest covered area ratio and add it to a flood-fill priority queue. After a cell is popped from the queue and subdivided, its interesting neighbors are added to the queue, with priority given to the lowest uncovered areas.
We repeat this process until no interesting cells remain without subdividing.
Within each cell to be subdivided, we compute the covered area ratio of each of its children, and assign them new spheres in ascending order of this ratio
(see \autoref{fig:subdivision-orders}).

\subsection{Further performance improvements}
\label{sec:performanceimprovements}

\changed{%
The DOS data structure significantly improves the runtime performance over the naive \texttt{grow\_to\_points} and \texttt{interpolate\_sdf} algorithms.
We employ a few minor improvements that further boost performance in practice (though not asymptotically).
}{furtherimprovements1}

\changed{%
If the interpolation point \(\point\) is within an existing sphere, we know that the sign of \(\point\) must match the \changed{sign of the existing sphere}{saxon-3}, and we can skip checking all opposite-sign candidate radii.
Before even computing the first candidate grow-to point, we can compute a lower and upper bound on the final radius using \newchanged{the following observations}{}:
no valid grow-to point can be closer than any tangent point of an existing sphere that contains \(\point\) (otherwise the growing sphere would be completely covered); no valid grow-to point can be farther than any tangent point of an existing opposite-sign sphere (otherwise the growing sphere would intersect the opposite-sign sphere); and no valid grow-to point can be farther away than both tangent points of any existing sphere, as in that case that existing sphere would have no uncovered point.
}{furtherimprovements2}
\begin{figure}
\centering
\includegraphics{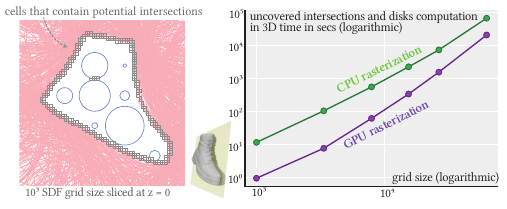}
\caption{\changed{%
To compute the initial uncovered intersections of all spheres, we rasterize the spheres to identify all pixels with more than two or three sphere contours, significantly reducing the number of sphere pairs and triples to check.
We further speed up this rasterization by performing it on the GPU.
}{furtherimprovements6}}
\label{fig:gpu-vs-cpu-rasterization}
\end{figure}

\newchanged{%
When subdividing, we can further reduce the number of relevant spheres to check for each child cell by employing a maximal coverage range $r_k$ for the child at depth \(k\). Let $\diagonal_{\textrm{k}}$ be the diagonal length of the child cell and \(|\sdfsample_{k-1}|\) be its parent sphere radius at depth \(k-1\). Then \(r_k=\max\left(|\sdfsample_{\textrm{k-1}}| + \diagonal_{\textrm{k}}/2, \diagonal_{\textrm{k}}\right)\), which we prove in \autoref{app:additionalproofs}, \autoref{prop:maximal-coverage-range-depth-k}. The first term, $|\sdfsample_{\textrm{k-1}}| + \diagonal_{\textrm{k}}/2$, is the largest radius a child sphere could have without definitely invalidating the parent sphere and while fully covering its cell. The second term, $\delta_k$, is the farthest distance that the farthest point of any subsequent child sphere (on any finer subdivision level) could be from the center of the cell (see \autoref{fig:max-child-radius}, the dashed yellow range).
}{furtherimprovements3}

\changed{%
Instead of checking whether the SDF is valid for each candidate radius starting from the highest-scoring one, we employ a variety of shortcuts.
We do not check whether all spheres fulfill the conditions of \autoref{prop:discrete-sdf-validity}, but only the existing spheres that are close enough to \(\point\) to matter (and the growing sphere).
We first check whether any existing spheres' uncovered intersection points lie outside of the growing sphere, since if they do, that existing sphere is still valid;
if not, we recompute the intersections of that sphere with the growing sphere only and check whether any of those are uncovered.
If all existing spheres that matter are still valid, we secondly move on to determining whether the growing sphere is valid: if the growing sphere hits an existing sphere's uncovered intersection point, we know that the growing sphere is valid (it has at least that uncovered point); should this not be the case we check whether the sphere created any new uncovered intersections points, and if this also is not the case the growing sphere is invalid.
}{furtherimprovements4}

\begin{figure*}
\centering
\includegraphics{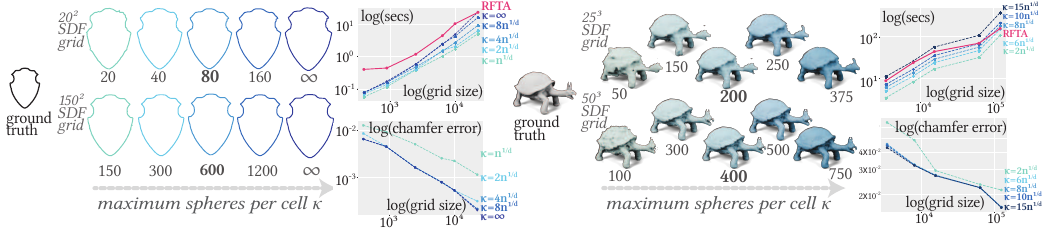}
\vspace{-0.7cm}
\caption{Culling spheres until we have \(\kappa\) relevant spheres per cell for refinement and reconstruction improves speed without degrading quality much. For a comparable runtime to RFTA as well as high quality reconstruction in 3D, we choose \(\kappa = 8n^{1/3}\) as default.}
\label{fig:culling}
\end{figure*}

\begin{figure}
\centering
\includegraphics{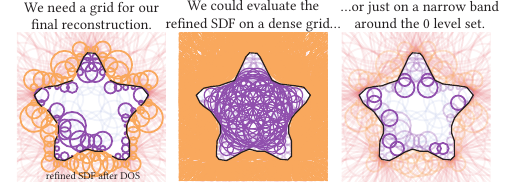}
\vspace{-0.7cm}
\caption{After refining the SDF with our DOS, we use marching cubes to reconstruct a surface, which requires samples on a uniform grid.
Reconstructing the DOS output (black, \emph{left}) exhibits gaps for cells missing some corner samples.
Instead of \newchanged{upsampling}{} SDF to a dense uniform grid \emph{(center)}, we extend it to a narrow uniform band around the 0 level set by computing \(\rmin\) at all relevant grid points \emph{(right)}.}
\label{fig:naive-vs-smart-upsampling}
\end{figure}

\changed{%
We do not consider all partially uncovered intersection circles in \(d=3\) for refinement and reconstruction, but only the completely uncovered ones, in order to cut down on the total number of circles to check.
This has not caused us to run into any problems in practice, but means that occasionally there could be a higher-scoring radius that is missed.
For pseudo-SDF repair, we do consider all partially uncovered disks in order to guarantee that the resulting SDF is globally consistent despite parallelization.
}{furtherimprovements5}

\changed{%
Updating uncovered intersection points during DOS refinement is cheap due to the small number of spheres involved, even though the complexity is \(O(k^{d+1})\) for \(k\) spheres.
On DOS initialization, however, this can become unacceptably large, since we have to consider all \(n\) spheres at once.
We speed up this intersection computation by rasterizing the contour of the union of all spheres.
Each pixel with two or more spheres has a potential intersection for \(d=2\) or uncovered intersection circle for \(d=3\), and each pixel with three or more spheres has a potential intersection for \(d=3\).
A k-d tree is used to find out which sphere contributes to which pixel, and we then perform intersection tests only on this extremely reduced number of sphere pairs and triplets (see \autoref{fig:gpu-vs-cpu-rasterization}, \emph{left}).
Following \citet[Fig.~10]{rfta}, rasterization performance is further improved using the GPU (see \autoref{fig:gpu-vs-cpu-rasterization}, \emph{right}).
Our default GPU rasterization resolution is \(10 n^{\frac{1}{d}}\) pixels along one axis, rounded up to the nearest power of \(2\) that is \(\geq 64\).
Even though our \texttt{interpolate\_sdf\_to} procedure should always find a valid radius for every point it is called at, due to numerical issues or floating point round-off errors, it might happen that we need to resort to a fallback procedure in very rare cases.
In this case, we run \texttt{min\_valid\_radius}, which is less vulnerable to floating point issues as it performs no SDF validity checks.
}{furtherimprovements7}

\begin{figure}
\centering
\includegraphics{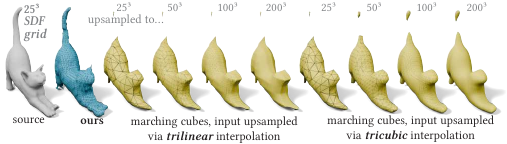}
\vspace{-0.7cm}
\caption{Our \changed{interpolation}{} approach differs from simply performing traditional interpolation (e.g., trilinear, tricubic) of the SDF data onto a finer grid.
We use \emph{global} information contained in an input SDF, which recovers a highly detailed surface, preserving features such as the tail and ears of the cat.}
\label{fig:comparison-with-mc-interpolation}
\end{figure}

\subsubsection*{Bounding the number of spheres per cell}
\changed{%
To improve the runtime of our method even more, we employ a \emph{culling} procedure.
As the runtime is polynomial in the number of relevant spheres per cell,
we cull spheres from the input until the largest number of relevant spheres per cell is \(\leq \kappa\), starting with the largest sphere of the cell with the most relevant spheres.
The output SDF is still valid with respect to all input spheres that were not culled, but can violate culled input spheres.
For reconstruction, our default value of \(\kappa=4 \cdot n^{1/2}\) for \(d=2\) and \(\kappa=8 \cdot n^{1/3}\) for \(d=3\) does not lead to a noticeable decrease in quality, but a significant increase in speed (see \autoref{fig:culling}).
}{culling-moved-up}

\section{Application: Reconstruction}
\label{sec:application-reconstruction}

The most commonly used methods to reconstruct triangle meshes from discrete SDF data, like Marching Cubes \cite{Lorensen1987} and Dual Contouring \cite{Ju2002}, are \emph{local} strategies that individually process the information contained in each cell to output a local approximation of the surface in that cell.
These methods are generally robust and computationally efficient.
Recently, \citet{rfts} noted that useful distance information is not restricted to individual cells; instead, samples arbitrarily far from the surface can also contain relevant information about it.
This insight generated a surge of interest in \emph{global} reconstruction methods \cite{rfta,kohlbrenner2025polyhedral,kohlbrenner2025isosurface} that seek to exploit all information from every SDF sample, albeit with significant impact on computational cost and robustness.
This context highlights the unique benefit of the refinement strategy introduced in \autoref{sec:application-refinement}.
The set of SDF samples in the DOS at its finest depth is a high-resolution narrow band of new SDF data in the local neighborhood of the surface \emph{that is consistent with all input SDF samples},
i.e., these are new \emph{local} samples that nevertheless properly reflect the \emph{global} information from all of the input's far-away signed distance samples.
Critically, this means that we can pair our refinement strategy with an efficient, robust local reconstruction method (e.g., Marching Cubes) to produce a reconstruction that inherits the accuracy of global methods with the desirable algorithmic properties of local ones.
This is different from applying a naive interpolation method and then using Marching Cubes: Our refinement prior to reconstruction ensures all global information is used (see \autoref{fig:comparison-with-mc-interpolation}).

\begin{wrapfigure}[14]{r}{0.3\columnwidth}
    \vspace*{-12pt}
    \hspace*{-20pt}
    \includegraphics[width=1.3\linewidth]{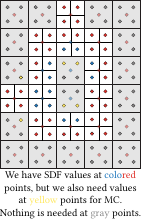}
\end{wrapfigure}
However, the lowest subdivision level of our DOS is not a complete narrow band containing the surface.
In practice, there will be some vertices in this narrow band for which we have not computed a radius, since they belong to \emph{uninteresting}
cells that were not subdivided earlier (see inset).
Hence, we identify these points, compute their minimum valid distance $\rmin$, and pass the entire sparse narrow band to Marching Cubes.
Using a sparse narrow band significantly speeds up the computation compared to using a dense grid (see \autoref{fig:naive-vs-smart-upsampling}). Furthermore, as we discuss below, these minimum valid distances can be computed efficiently in parallel, \changed{which makes this postprocessing step efficient to compute}{explaining-parallel-4}.

\begin{figure}
\centering
\includegraphics{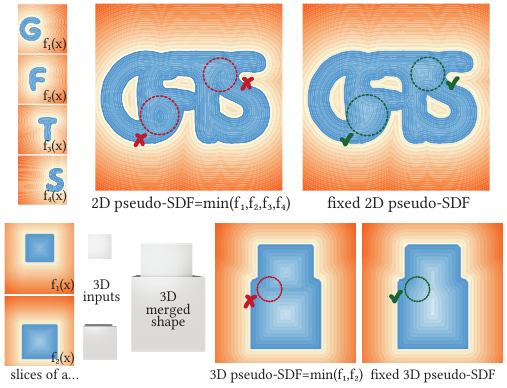}
\vspace{-0.7cm}
\caption{Given a pseudo-SDF (left), our method can return a repaired, valid SDF (right) by computing the minimum valid radius of a sphere at every invalid input sample.}
\label{fig:pseudo-sdf}
\end{figure}

\section{Application: pseudo-SDF repair}
\label{sec:application-pseudosdf}

The global SDF refinement strategy proposed in \autoref{sec:application-refinement} required each new signed distance to be assigned \emph{sequentially}, since a given choice of $\sdfsample_{n+j}$ can easily affect the set of valid choices for all subsequent points $\point_{n+j+1},\dots,\point_{m}$.
Computing the new spheres' radii must then be done one at a time, to ensure the new sphere radii remain consistent with the previously assigned ones. Thus, the loop in \autoref{sec:flood-fill-refinement} cannot (easily) benefit from parallelization.
However, a critical realization is that for the specific case in which the new points are assigned the \emph{minimum valid radius} constructed in \autoref{prop:validradii}, this consistency is guaranteed theoretically.

\begin{corollary}
    \label{prop:rmin}
    Let $\{(\point_{i},\sdfsample_i)\}_{i=1}^n$ be a set of valid signed distance samples, and consider a new set of points $\point_{n+1},\dots,\point_{m}$.
    Let $\rmin_{n+j}$ be the minimum value that makes $\{(\point_{i},\sdfsample_i)\}_{i=1}^n \cup (\point_{n+j},\rmin_{n+j})$ valid (which exists and can be calculated using \autoref{prop:validradii}). Then, the complete set
    \[
    \{(\point_{i},\sdfsample_i)\}_{i=1}^n \cup \{(\point_{n+j},\rmin_{n+j})\}_{j=1}^{m-n}
    \]
is valid.
\end{corollary}
\begin{proof}
\changed{%
    Let $C_\pm$ be the positive and negative contours of the input spheres $\{(\point_{i},\sdfsample_i)\}_{i=1}^n$.
    If a point $\point_{n+j}$ is outside an existing sphere, its minimum valid radius is $0$.
    As shown in \autoref{prop:validradii}, the minimum valid radius for each $\point_{n+j}$, $\rmin_{n+j}$ inside an existing sphere is the distance from $\point_{n+j}$ to the contour $C_{\pm}$ of the appropriate sign. This observation has two immediate consequences:
    \begin{enumerate}
        \item The sphere $(\point_{n+j},\rmin_{n+j})$ is tangent to the appropriately signed input sphere contour $C_\pm$. Therefore, it contains an intersection point $q$ that is not covered by the interior of any of the input spheres.
        \item The sphere $(\point_{n+j},\rmin_{n+j})$ is the minimal sphere centered at $\point_{n+j}$ that is tangent to $C_\pm$. Consequently, its interior does not contain any point on the contour $C_\pm$.
    \end{enumerate}
}{corollary-proof-1}

\changed{%
    For each $j$, the first of these properties ensures that the sphere $(\point_{n+j},\rmin_{n+j})$ contains an intersection point that is not covered by any of the input spheres. In turn, the second property (when applied to all other output spheres) ensures that this point is also not covered by any of the new spheres. Thus, $(\point_{n+j},\rmin_{n+j})$ contains an uncovered intersection point.
    Finally, the second property also ensures that all the input spheres remain valid after adding all $(\point_{n+j},\rmin_{n+j})$. Thus, all spheres $\{(\point_{i},\sdfsample_i)\}_{i=1}^n \cup \{(\point_{n+j},\rmin_{n+j})\}_{j=1}^{m-n}$ are valid.
}{corollary-proof-12}
\end{proof}

\begin{figure}
    \centering
    \includegraphics[]{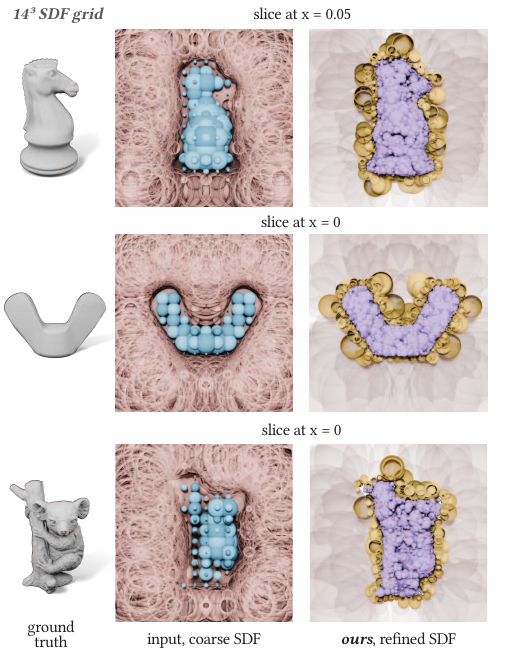}
    \vspace{-0.7cm}
    \caption{Our method refines a coarse SDF onto a finer grid in a way that avoids the pitfalls of the naive interpolation methods of \autoref{fig:fig2}.
    \newchanged{The SDFs are visualized as 2D slices through volumes populated with spheres: positive spheres are shown as cross-sections, while negative spheres are rendered in full. Refined samples (right) are highlighted in purple (negative) and yellow (positive).
}{new-caption}}
    \label{fig:refinement}
\end{figure}

The value of $\rmin$ can be computed even more efficiently: since it is the minimum valid radius, we can ignore the possibility that the new sphere covers a previously uncovered sphere and only look for the minimum radius in $P$ that causes the growing sphere to have an uncovered point.
\changed{%
Our algorithm for finding \(\rmin\) is:
}{newalg-smallestrad}
\begin{lstlisting}
//Return the smallest possible valid radius for a sphere at this point.
function min_valid_radius_for_known_sign(p, sign, p1, ..., pn, s1, ..., sn):
  rmin = inf
  //Find the closest possible uncovered intersection. Since this is the smallest valid radius, no chance of covering any other sphere, so no SDF validity checks needed.
  for i = 1,...,n:
    //Skip all checks with opposite-signed spheres, as we will remain within the same-signed contour.
    if sign*si < 0:
      continue
    for j = i+1,...,n:
      if d==2:
        if q1,q2 = sphere_intersect_2d(pi,|si|,pj,|sj|):
          for q = q1,q2:
            if point_uncovered(q,p1,...,pn,s1,...,sn):
              rmin = min(rmin, len(q-p))
      else if d==3:
        for k = j+1,...,n:
          if q1,q2 = sphere_intersect_3d(pi,|si|,pj,|sj|, pk,|sk|):
            for q = q1,q2:
              if point_uncovered(q,p1,...,pn,s1,...,sn):
                rmin = min(rmin, len(q-p))
        //In 3D, add circle points
        if zeta,rho,nu = intersection_circle_uncovered(pi, |si|,pj,|sj|):
          q1 = project_to_closest_pt_on_circle(p,zeta,rho,nu)
          q2 = project_to_farthest_pt_on_circle(p,zeta,rho,nu)
          for q = q1,q2:
            if point_uncovered(q,p1,...,pn,s1,...,sn):
              rmin = min(rmin, len(q-p))
    //Add tangent points, but only do uncovered check if they can possible be a closest point. As before, no SDF validity check needed.
    for i = 1,...,n:
      for c = -1,1:
        q = pi + c*normalized(p-pi)*|si|
        if len(q-p) < rmin:
          if point_uncovered(q,p1,...,pn,s1,...,sn):
            rmin = min(rmin, len(q-p))
    return sign*rmin
\end{lstlisting}

We now review an application for which only computing
$\rmin$ is sufficient, allowing us to make full use of \autoref{prop:rmin}, this improved efficiency, and parallelization: \emph{repairing} pseudo-SDFs by converting them into \changed{valid}{} SDFs without modifying the zero level set.
Pseudo-SDFs are often accurate near the surface's zero level set \cite{Marschner2023}, but contain entire regions away from it in which the distance to the surface is underestimated (see \autoref{fig:pseudo-sdf}).
This makes their repair an ideal candidate for the minimum distance refinement strategy enabled by \autoref{prop:rmin}.
Concretely, from a discrete set of pseudo-SDF samples, we propose first gathering a list of all fully covered spheres (i.e., whose SDF value is underestimated).
We remove the distance value stored for all of them, leaving us with a set of valid SDF samples and a separate set of points with no assigned distance.
We loop (in parallel) over this set, assigning \(\rmin\) to each point.
Since pseudo-SDFs are conservative  (i.e., absolute lower bounds of the true distance \cite{Marschner2023}) we need only consider radii larger than the point's originally assigned distance.
This results in a repaired, globally consistent set of \changed{valid}{true-to-valid} SDF samples.

\begin{figure}
    \centering
    \includegraphics[]{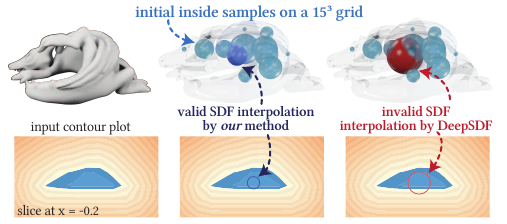}
    \vspace{-0.7cm}
    \caption{\changed{Given only a coarse SDF as input, our method guarantees valid SDF interpolation at arbitrary points in space, while neural approaches such as \citet{DeepSDF_CVPR} can produce contradictory SDF values.}{newfigure-nsdf-viol}}
    \label{fig:nsdf-violation}
\end{figure}

\begin{figure}
    \centering
    \includegraphics{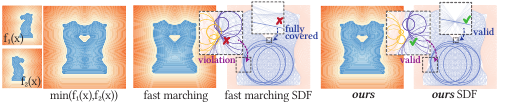}
    \vspace{-0.8cm}
    \caption{\changed{Fast approximate redistancing methods for pseudo-SDFs, such as the Fast Marching Method (applied to a linearly interpolated zero level set) \cite{scikit-fmm} don't guarantee the validity of the new SDF. Our method produces a valid SDF after repairing.}{new-figure-pseudo-fmm}}
    \label{fig:pseudo-fmm}
\end{figure}

\section{Experiments and results}
\label{sec:results}

We now validate our implementation choices, compare to prior work, and show applications of our method.
Details for the experiments in this section, such as implementation details, errors, timing, and parameter choices, can be found in the \changed{Appendix}{supptoapp}.
We plan to make our code publicly available after acceptance.

Our algorithm interpolates an SDF value at a given position in space and guarantees (unlike previous work) that the resulting SDF, with the addition of the new point, is valid.
Beyond this important theoretical contribution, our work has application in three areas of SDF processing that we highlight here: \emph{SDF refinement}, \emph{surface reconstruction}, and \emph{pseudo-SDF repair}.

\subsubsection*{Refinement}
We refine a coarse SDF in Figs. \ref{fig:teaser} and \ref{fig:redistancing} and show that previous methods output SDFs that violate the input SDF. \newchanged{Additional SDF refinement results are shown in \autoref{fig:refinement}}{refinement-results-added}.
\autoref{fig:nsdf-violation} shows that even recent neural methods which allow for interpolation of an SDF to any point do not guarantee a valid output SDF or compatibility with the input spheres.

We justify our choice of the greedy order in which we refine our SDF in \autoref{fig:subdivision-orders}, and the preference for inside spheres in our radius scoring procedure in \autoref{fig:radius-scores}.
Both refinement and reconstruction require the choice of a subdivision level \(\subdivisiondepth\), which we ablate in \autoref{fig:subdivision-levels}.

\begin{figure}
    \centering
    \includegraphics[]{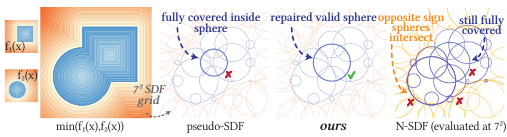}
    \vspace{-0.8cm}
    \caption{\changed{Unlike some neural methods for pseudo-SDF repair \cite{Marschner2023}, our method guarantees SDF validity after repairing.}{new-experiment}}
    \label{fig:nsdf-didactic}
\end{figure}

\begin{figure}
    \centering
    \includegraphics{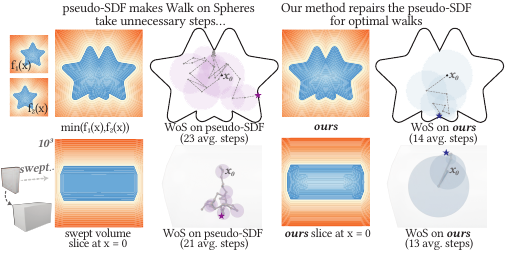}
    \vspace{-0.8cm}
    \caption{\changed{Walk-on-spheres on a pseudo-SDF result of a CSG operation on SDFs causes longer interior walks, whereas using our true repaired SDF enables faster arrival within the boundary \(\varepsilon\) (and hence termination).
    \emph{Top}: Union of two 2D stars.
    \emph{Bottom}: Swept volume of a 3D block.}{new-figure-wos}}
    \label{fig:wos}
\end{figure}

\subsubsection*{Pseudo-SDF repair}
Pseudo-SDF repair is a critical task with applications in industrial design, manufacturing and 3D modeling.
In Figs.~\ref{fig:teaser}, \ref{fig:pseudo-sdf}, we show that our algorithm can repair pseudo-SDFs, converting them into SDFs that are theoretically valid and can be used by downstream tasks (e.g., sphere tracing).
\changed{%
\autoref{fig:marschner} compares to one of the few prior works attempting this task \cite{Marschner2023}.
\autoref{fig:nsdf-didactic} specifically shows that their approach does not guarantee that the repaired SDF is consistent with the input pseudo-SDF, while ours is consistent.
}{}
While \autoref{fig:marschner}'s setup is slightly different (they output a smooth function with much higher theoretical resolution), for the case of discrete samples, we manage to repair pseudo-SDFs much more efficiently.
\changed{\autoref{fig:pseudo-fmm} shows that merely using the Fast Marching method \cite{sethian1996fast} for redistancing the zero level set of a pseudo-SDF is not enough and can lead to SDF violations.}{new-pseudo-fmm-text}

\begin{figure}
    \centering
    \includegraphics[]{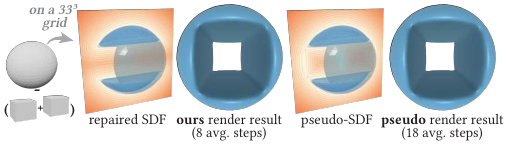}
    \vspace{-0.7cm}
    \caption{\newchanged{Using sphere tracing to render a pseudo-SDF (right) constructed as the difference of two SDFs (left) on a $33^3$ grid using trilinear interpolation requires more iterations and time (see \autoref{tab:sphere-tracing}) than rendering a repaired SDF (middle) due to incorrect exterior distance values.}{new-figure-sphere-tracing}}
    \label{fig:sphere-tracing}
\end{figure}

\subsubsection*{\newchanged{CSG operations and offsetting}{}}

An important use of pseudo-SDFs is in CSG operations.
It is easy to compute Boolean operations on SDFs by simply taking maxima and minima:
the zero level set of \(\min(\sdf_1, \sdf_2)\) is the union of the two shapes, even though the resulting function is only a pseudo-SDF.
SDFs are also extremely useful for offsetting operations:
a surface offset by \(\omega\) is simply the SDF meshed at the level set \(\omega\) instead of \(0\).
For this basic CSG operation to work, however, it is imperative that the SDF is a \emph{valid} SDF.
Invalid SDFs such as the pseudo-SDFs resulting from simple \(\min\)-unions generate wrong offsets, as can be seen in \autoref{fig:offset-3d}, where we take the offsets of unioned objects (the offset of a pseudo-SDF is, wrongly, the offsets of the individual shapes before unioning).
Repairing the SDFs with our method guarantees that the offset is computed correctly.

\subsubsection*{\newchanged{Applications requiring distance queries}{}}

SDFs are also a useful geometric representation for solving PDEs with Monte Carlo \emph{Walk on Spheres} methods \cite{Sawhney2020,Sawhney2023}.
In order to guarantee that the sphere walks from an evaluation point to the boundary are as short as possible, it is imperative that the SDF validly describes the true distance to the boundary at any point (otherwise the walk takes too many steps).
Repairing a pseudo-SDF into a valid SDF with our method thus significantly speeds up Monte Carlo PDE solvers on implicit functions resulting from CSG operations, as in \autoref{fig:wos}.
Implicit functions like SDFs are also used in rendering, and can be visualized with \emph{sphere tracing} \cite{Hart1996,Liu2020}.
While sphere tracing can be used for more general implicit surface types, including pseudo-SDFs, it is most efficient for true distance fields; otherwise the number of required evaluations along a single ray goes up significantly.
Repairing a pseudo-SDF with our method drastically cuts the number of evaluations required as shown in \autoref{fig:sphere-tracing}. \newchanged{We also evaluate the total rendering runtime accounting for the cost of field evaluations.
Sphere tracing on the repaired SDF takes fewer steps than directly evaluating the pseudo-SDF as soon as there are more than a few steps per ray, and this difference accumulates as the number of rays increases.
For a sufficiently large number of rays, the accumulated runtime savings exceed the one-time cost of repairing the SDF, making the repair beneficial for applications involving repeated distance queries (see \autoref{tab:sphere-tracing}).}{sphere-tracing-writing}
Thus, while both Monte-Carlo-solving and sphere tracing technically accept pseudo-SDFs, they become inefficient and profit from having exact distance data available to them.

\begin{figure}
    \centering
    \includegraphics{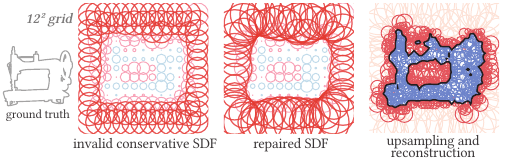}
    \vspace{-0.7cm}
    \caption{\newchanged{Our applications can be extended to conservative SDFs. Having detected an invalid SDF, we can first apply pseudo-SDF repair, and then use upsampling for better reconstruction results.}{conservative-sdf-figure}}
    \label{fig:conservative-sdf}
\end{figure}
\subsubsection*{Reconstruction}
In Figs.~\ref{fig:teaser}, \ref{fig:comparison-with-mc-interpolation} we show that our reconstruction method extracts global information, producing detailed reconstructions of coarse SDFs that cannot achieved with local reconstruction methods, even if they employ high-order naive interpolation of SDF values onto fine grids.
\newchanged{In all of our comparisons to \emph{plain} Marching Cubes (without naive interpolation), the resolution of our reconstruction will always be larger than Marching Cubes', because we refine the SDF before reconstruction.}{}

\autoref{fig:rfta-prior} shows that while some previous methods \changed{tend to oversmooth reconstructed surfaces}{}, this is not the case for our method, which has no trouble reconstructing objects with both smooth and non-smooth elements.
\autoref{fig:complicated-geometry} shows that even very complicated geometries are reconstructed with a high degree of quality.
Figs.~\ref{fig:reconstruction-comparison-3d}, \ref{fig:reconstruction-comparison-2d} show large-scale qualitative comparisons of our method with the state-of-the-art in global SDF reconstruction, which demonstrate that our method does not sacrifice visual quality for the guarantees that it brings, while 
\autoref{fig:big-dataset-scatterplot} shows a quantitative comparison \changed{(error metrics can be found in \autoref{tab:runtimes-scatterplot})}{runtimes-ref-added-2}.
We observe that our method matches competing methods in numerical accuracy, while providing stronger guarantees of SDF consistency.

\begin{figure}
\centering
\includegraphics{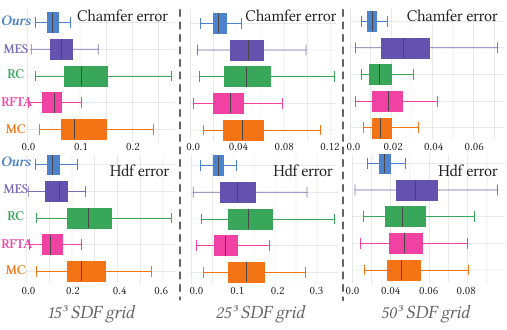}
\vspace{-0.7cm}
\caption{\changed{On 100 random shapes from the TetWild dataset \cite{hu2018tetrahedral} across several resolutions, our algorithm matches the accuracy of the best global SDF reconstruction methods while guaranteeing adherence to the input distance data.
Number of failures: \(15^3\): MC 3, RC 4, MES 13; \(25^3\): MES 4.
Runtimes and errors can be found in \autoref{tab:runtimes-scatterplot}.}{runtimes-ref-added-1}}
\label{fig:big-dataset-scatterplot}
\end{figure}

\begin{figure*}
    \centering
    \includegraphics{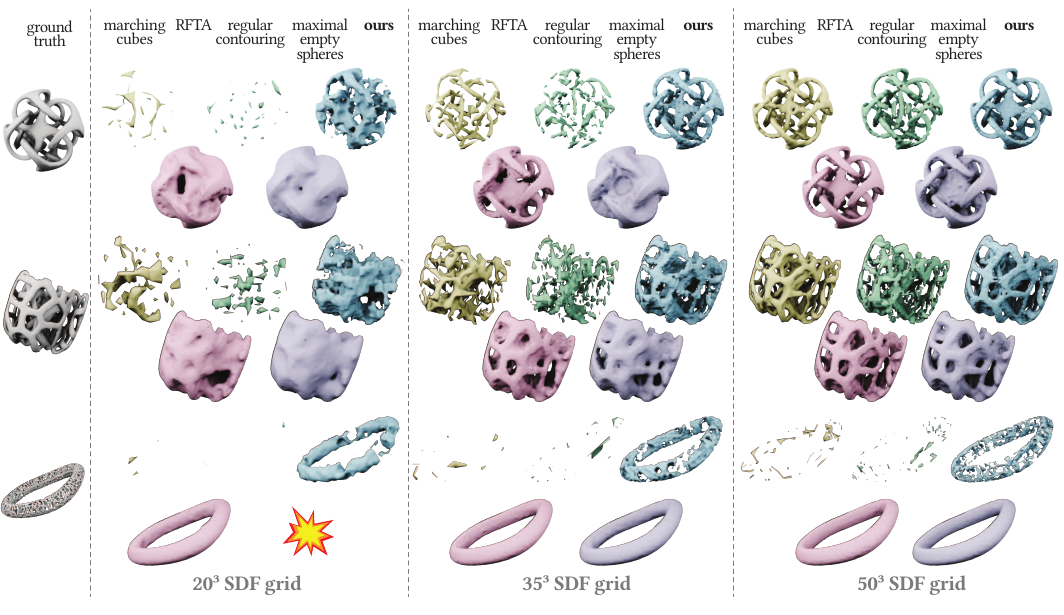}
    \caption{Our method effectively reconstructs complex geometric structures \newchanged{with intricate features and multiple layers}{new-figure} from SDFs, even where previous work struggles \newchanged{or fails}{new-figure} to capture all the fine details.}
    \label{fig:complicated-geometry}
\end{figure*}

\begin{figure}
\centering
\includegraphics[width=243.14749pt]{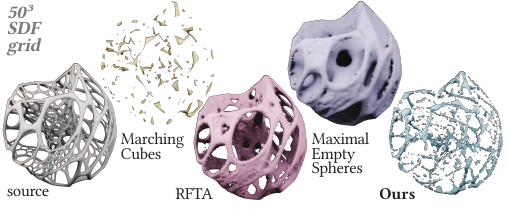}
\vspace{-0.6cm}
\caption{Our method relies on adequate initial interior SDF samples for successful reconstruction. \newchanged{For complex geometries containing extremely thin branches or fine structures relative to the grid resolution, our method can fail to recover these features (and previous methods produce other kinds of artifacts).}{new-figure}}
\label{fig:limitations}
\end{figure}

\changed{Thanks to the performance improvements described in section \autoref{sec:performanceimprovements}, our method's runtime either matches or is slightly longer than that of the most accurate discrete coarse SDF reconstruction methods while, unlike them, guaranteeing the consistence of the generated SDF samples. Average runtimes are presented in \autoref{tab:runtimes-scatterplot}, and runtimes for each individual figure are shown in \autoref{tab:runtimes-2dcomparison} and \autoref{tab:runtimes-3dcomparison}.
}{runtimes}
\newchanged{Notably, our method targets a problem setting in which the discrete SDF data make up the entirety of the input, without access to the true SDF function and thus without the ability to draw any additional SDF samples.
As shown in \autoref{fig:big-dataset-scatterplot}, our algorithm is the preferred choice in this setting; if one does have access to the true SDF function, sampling it on a denser grid or octree and using more efficient reconstruction methods is likely preferable.\label{changed:we-can-not-resample}}{}

\section{Limitations \& Conclusion}
\label{sec:conclusion}
\newchanged{We have introduced the first method for validity-preserving SDF refinement, a method for reconstruction with strong detail preservation, and a method for pseudo-SDF repair which achieves orders-of-magnitude speedups over existing approaches.}{}
\newchanged{Meanwhile, the theory we have presented lays out the essential mathematical requirements that any valid interpolation strategy for discrete SDFs must satisfy.}{}

Since our method greedily grows, the order in which we subdivide matters.
 In future work, we hope to explore methods that can grow all new spheres at the same time, resulting in an order-independent method, and even the ability to explicitly include certain priors as secondary objectives during optimization;
this will also address the fact that our method struggles if there are very few negative SDF values in the input (see \autoref{fig:limitations}), and can sometimes produce surfaces that are very rough (and thus might profit from an explicit, controllable smooth prior).

We carried out the walk-on-spheres and sphere tracing tasks to demonstrate that performing pseudo-SDF repair can yield downstream efficiency benefits
\newchanged{(although we currently require a large number of sphere tracing rays to overcome the cost of precomputation, see \autoref{tab:sphere-tracing}).
As future work improves the performance of our method (e.g., with a GPU implementation), we hope that these efficiency benefits compound.}{sphere-tracing-limitations}
These schemes could also benefit from our SDF refinement strategy: just as marching cubes applied to a refined SDF yields greater accuracy, refinement should likewise improve the quality of walk-on-spheres and sphere tracing results.

Our method is the first to impose SDF validity as a series of \newchanged{hard constraints}{}.
As a direct consequence, if the input violates strict SDF validity (for example, because it is noisy, \newchanged{see \autoref{fig:noise}}{}), our method can fail to find a valid SDF value for a point.
Interesting future work is the extension of our approach to support noisy input SDFs that contain violations beyond what is allowed by our pseudo-SDF repair method, e.g., opposite-sign sphere intersections.
\newchanged{Separately, our work focuses on exact SDF validity, whereas small, controlled violations may be acceptable for some downstream applications. Future work could explore the challenges and potential benefits of relaxing our hard constraints to allow such approximate solutions.}{}

Finally, our method relies on the specific mathematical properties of exact \newchanged{SDFs}{} to evaluate and propose consistent new distance data. \newchanged{Other classes of implicit representations not based on exact SDFs}{} like BlobTrees \cite{wyvill1999extending}, convolution surfaces \cite{bloomenthal1991convolution}, R-functions \cite{pasko1995function}, SCALIS models \cite{zanni2013scale} and occupancy fields \cite{mildenhall2021nerf} \newchanged{may technically be valid inputs for our repair algorithm, but the resulting repaired SDF is not necessarily useful.}{new-limitations}
Deriving a general-purpose class of interpolants that can generate implicit function values consistent with arbitrary (potentially spatially mixed) sources of data is a promising avenue for impactful future work.

\begin{figure}
    \centering
    \includegraphics[]{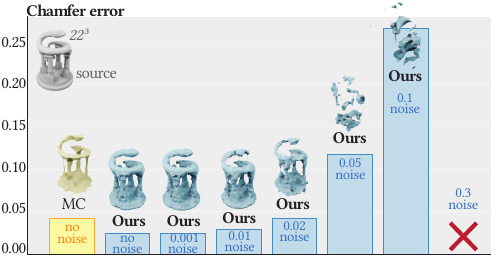}
    \vspace{-0.4cm}
    \caption{\changed{Adding Gaussian noise to the SDF input values,
with increasing standard deviation. Our reconstruction is relatively robust to moderate noise in the input SDF, even outperforming MC (with noise-free input) for up to a deviation of 0.02. However, at large deviations, our algorithm cannot find a valid SDF value, which leads to unreasonable reconstruction results.}{newfigure3}}
    \label{fig:noise}
\end{figure}

\begin{acks}
We thank David Cha for proofreading.
We thank Zoë Marschner, Maximilian Kohlbrenner, Pranav Jain and Dylan Rowe for help with coding.
ChatGPT and Copilot were utilized to generate and debug parts of the code.
We gratefully acknowledge the creators of the 2D and 3D assets used in this article. Figures in this work contain the moon \cite{moon-figure}, walking cat \cite{walking-cat-figure}, star \cite{star-figure}, angkor \cite{angkor-figure}, armbrost \cite{armborst-figure}, cat \cite{cat-silhouette-figure}, dragon \cite{blue-dragon-figure}, fishbone \cite{fish-skeleton-figure}, maple leaf \cite{maple-leaf-figure}, zombie \cite{cartoon-zombie-figure}, polytree \cite{low-poly-tree-sculpture-mesh}, angel candle holder \cite{angel-candle-holder-mesh}, david \cite{david-head-mesh}, hat \cite{chopper-hat-mesh}, fandisk \cite{fandisk-mesh}, spot \cite{spot-mesh}, tower \cite{two-tower-mesh}, cloud \cite{cloud-mesh}, cactus \cite{cactus-mesh}, dragon tower \cite{tower-mesh}, turtle pope \cite{turtle-pope-mesh}, boot \cite{boot-mesh}, cat \cite{cat-mesh}, wingnut \cite{wingnut-mesh}, springer \cite{springer-mesh}, koala \cite{koala-mesh}, pet monster \cite{pet-monster-mesh}, orange \cite{fruit-bunch-mesh}, metratron \cite{metratron-mesh}, lace rose \cite{Lace-Digital-Rose-mesh}, mug \cite{mug-cup-mesh}, statue \cite{snake-statue-mesh}, splatter \cite{splattershot-mesh}, crab \cite{crab-mesh}, dragon skull \cite{dragon-skull-mesh}, lion statue \cite{lion-mesh}, monster \cite{Elephanticus-mesh}, well \cite{well-mesh}.

This research was supported by a gift from Adobe Inc and the National Science Foundation (award \#2335493).
This work was supported by the Natural Sciences and Engineering Research Council of Canada (Grant RGPIN-2021-02524).
The Geometry and the City lab at Columbia University is supported by generous gifts from nTop, Adobe, Dandy and Braid Technologies.
\end{acks}

\bibliographystyle{ACM-Reference-Format}
\bibliography{bibliography}

\appendix
\section*{Appendix}

\begin{newchangedblock}
    
\section{Additional Proofs}\label{app:additionalproofs}

\begin{proposition}[True Interpolation]\label{prop:trueinterpolation}
Let \(\point \to \point_i\). Then every valid SDF value \(\sdfsample\) at \(\point\) converges to the input value \(\sdfsample_i\). In particular,
$$
\lim_{\point \to \point_i} \sdfsample = \sdfsample_i
$$
for any \(\sdfsample\) corresponding to a point in \(\hat P\).
\end{proposition}

\begin{proof}
By \autoref{prop:validradii}, for every \(\point\), the set of valid grow-to points \(\hat P\) is non-empty, so at least one valid SDF value exists at \(\point\).
For any valid SDF value \(\sdfsample\) at \(\point\), the distance property requires every pair of samples to satisfy
$$
\left|\sdfsample-\sdfsample_i\right|
\leq
\left\|\point-\point_i\right\|.
$$
As \(\point\to\point_i\), the right-hand side tends to zero, and hence
$$
\lim_{\point\to\point_i}\left|\sdfsample-\sdfsample_i\right|=0.
$$
Thus every valid interpolated value converges to the original value \(\sdfsample_i\).
Because our method returns a valid SDF value, it thus performs true interpolation of the input samples.
\end{proof}
\end{newchangedblock}
\begin{newchangedblock}
    \begin{proposition}[Maximal Coverage Range]\label{prop:maximal_coverage_range} Let \(c\) be an interesting cell (as defined in \autoref{sec:DOS}) with associated known sphere \((\point,\sdfsample)\). The \emph{maximal coverage range}, defined as the radius of the bounding sphere centered at $\point$ that is guaranteed to cover all possible child spheres of $c$ (i.e., no child sphere can extend beyond this range), is \(r =\max\left(|\sdfsample| + \diagonal/2, 3\diagonal/4\right)\), where \(\diagonal\) is the length of the diagonal of \(c\).
    \end{proposition}
    \begin{proof}
        To determine the range of the radii of the children of $c$, we first consider subdivision depth 1, for which there are two cases.
        
        \indent Case 1 (\autoref{fig:max-child-radius}, left):
        The child sphere \((\point_1,\sdfsample_1)\) (plum pink in \autoref{fig:max-child-radius})
        fully covers its own cell. In this case, the child cell is not interesting so no further subdivision is performed. The maximum radius
        of the child sphere occurs when the child sphere contains and lies tangent to the parent sphere (\(\point,\sdfsample\)) (pale pink), hence \(|\sdfsample_1| =  |s| + \|p-p_1\| = |s| + \delta/4\). To contain this child sphere the radius of a bounding sphere (dashed red) centered at $\point$ must be \(r = |\sdfsample_1| + \|p-p_1\| = |\sdfsample_1| + \delta/4 = |s| + \delta/2\), which is the desired coverage range.

        \indent Case 2: The child (\(\point_1,\sdfsample_1\)) at depth 1 does not fully cover its own cell, so it is interesting and further subdivision at depth 2 is allowed. In this case, the maximal radius of the child at depth 1 is \(|\sdfsample_1| = \delta/4 \) (i.e., at most grazing the parent cell's corner). The child sphere at depth 2 (\(\point_2,\sdfsample_2\)) (light purple in \autoref{fig:max-child-radius}) also has two cases:

        \begingroup
        \leftskip=2em
        \noindent
        Case 2.1 (\autoref{fig:max-child-radius}, middle): (\(\point_2,\sdfsample_2\)) fully covers its own cell so it is uninteresting and subdivision halts. The maximal radius is \(|s_2| = |s_1| + \delta/8= \delta/4 + \delta/8 = 3\delta/8\), when the depth 2 child covers and lies tangent to the depth 1 child; the coverage range (dashed red) relative to $\point$ is \(r = |s_2| + |p_2-p_1| + |p_1-p| = |s_2| + \delta/8 + \delta/4 = 3\delta/4\).\\
        Case 2.2 (\autoref{fig:max-child-radius}, right): (\(\point_2,\sdfsample_2\)) does not fully cover its own cell so subdivision can continue to depth 3. In this case the, maximal radius is \(|\sdfsample_2| = \delta/8\). Again following the same logic, if the child sphere at depth 3 (dark blue in \autoref{fig:max-child-radius}) fully covers its own cell, the maximal radius \(|\sdfsample_3| =  3\delta/16\), and the coverage range (dashed red) \(r = |\sdfsample_3| + \delta/16 + \delta/8 + \delta/4 = 5\delta/8\).
        \par
        \endgroup

        \indent We want to show that in Case 2 the coverage range \(r\) continues to decrease as subdivision depth increases. At an arbitrary depth \(k\), the coverage range \[r = |\sdfsample_k|\ +\textrm{ displacement from } \point_k \textrm{ to }  \point.\]

        \indent The diagonal length of the cell at depth \(k\) is \(\delta/2^k\). In order to reach depth \(k\), the sphere at depth \(k-1\) must not fully cover its cell, hence \(|s_{k-1}| \leq \delta_{k-1}/2 = \delta_k = \delta/2^k\). Therefore \(|s_k| \leq |s_{k-1}| + \delta_k/2 = 3\delta/2^{k+1}\), which is the maximal radius that does not fully cover \((\point_{k-1},\sdfsample_{k-1})\).

        \indent The displacement from \(\point_k\) to \(\point\) is \(\sum_{j=1}^k \delta_j/2\). Combining the two parts, we have 
        \begin{align*}
            r & = |s_k| + \sum_{j=1}^k \frac{\delta_j}{2}\\
            & \leq \frac{3\delta}{2^{k+1}} + \frac{\delta}{2}\left(1-\frac{1}{2^k}\right)\\
            & = \left(\frac{1}{2} + \frac{1}{2^k}\right)\delta.
        \end{align*}

        \indent As \(k\) increases, \(1/2^k\) decreases, hence the maximal coverage range in Case 2 is when \(k=2\), i.e., \(r = 3\delta/4 \).

        Combining the results of Case 1 and Case 2, we have the maximal coverage range \[
            r = \max\left(
            |\sdfsample| + \diagonal/2,\,
            3\diagonal/4
            \right).
            \]
        
    \end{proof}

    \begin{proposition}[Maximal Coverage Range at Depth k] \label{prop:maximal-coverage-range-depth-k}
    Let \(c_k\) be an interesting cell with associated, yet-to-be-grown sphere \((\point_k,\sdfsample_k)\) at depth \(k\). The maximal coverage range at depth \(k\), defined as the radius of the bounding sphere centered at \(\point_k\) that is guaranteed to cover the SDF sphere of \(c_k\) as well as all of its possible children's spheres, is \ \(r_k =\max\left(|\sdfsample_{k-1}| + \diagonal_k/2, \delta_k\right)\), where \(\delta_k\) is the length of the diagonal of \(c_k\) and \(|\sdfsample_{k-1}|\) is the radius of its parent sphere at depth \(k-1\).
    \end{proposition}
    \begin{proof}
        There are two cases for \(r_k\).

        \indent Case 1 (\autoref{fig:max-child-radius}, left): \((\point_k,\sdfsample_k)\) fully covers its cell \(c_k\). In this case, \(c_k\) is no longer interesting so no more subdivision is performed. The maximum radius of the sphere occurs when it contains and lies tangent to the parent sphere \((\point_{k-1},\sdfsample_{k-1})\) (in \autoref{fig:max-child-radius} we use \(k=1\) for illustration, so the parent sphere is \((\point,\sdfsample)\)), hence the coverage range (dashed yellow) centered at \(\point_k\) is \(r_k = |s_k| = |s_{k-1}| + \|\point_k - \point_{k-1}\| = |s_{k-1}| + \delta_k/2\).

        \indent Case 2 (\autoref{fig:max-child-radius}, center, right): \((\point_k,\sdfsample_k)\) doesn't fully cover cell \(c_k\), so subdivision at depth \(k+1\) can be performed. The same idea follows as Case 2 in \autoref{prop:maximal_coverage_range}. The coverage range at depth \(k\) (dashed yellow) is 
        \[r_k = |s_{k+t}| +\textrm{ displacement from } \point_{k+t} \textrm{ to }  \point_k \]
        where \(k+t\) is an arbitrary depth further than \(k\), i.e., \(t \geq 1\). By \autoref{prop:maximal_coverage_range} Case 2, we know that \(r_k\) also decreases as \(t\) increases, since 
        \[r_k = (1/2 + 1/2^t)\delta_k.\]
        Therefore the maximal coverage range at depth \(k\) in Case 2 occurs when \(t=1\), i.e., \(r_k = \delta_k\).

        Combining the results of Case 1 and Case 2, we have the maximal coverage range at depth \(k\),
        \[r_k = \max\left(|\sdfsample_{k-1}| + \diagonal_k/2, \delta_k\right).\]
    \end{proof}

\end{newchangedblock}

\begin{figure}
    \centering
    \includegraphics{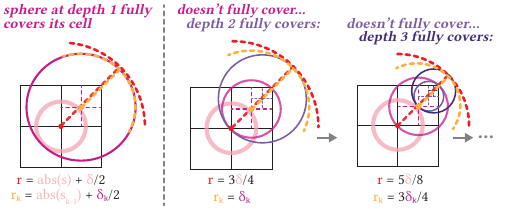}
    \vspace{-0.7cm}
    \caption{\newchanged{For an initial cell, the maximum range covering all child spheres is \(r =\max\left(|\sdfsample| + 0.5\diagonal, 0.75\diagonal\right)\) (red), where \(\delta\) is the diagonal of the cell, and \(|\sdfsample|\) is the radius of the initial sphere of the cell. At each subdivision depth \(k\), the newly refined child sphere has its own maximum radius range 
    \(r_k = \max\left(|\sdfsample_{\textrm{k-1}}| + 0.5\diagonal_{\textrm{k}}, \diagonal_{\textrm{k}}\right)\) (yellow), where \(\diagonal_{\textrm{k}}\) is the diagonal of its associated cell and \(|\sdfsample_{\textrm{k-1}}|\) is the radius of its parent sphere. Here we illustrate the depth-1 subdivision range in yellow, in this case \(|\sdfsample_{\textrm{k-1}}| = |s|\). }{max-valid-radius-figure-added}}
    \label{fig:max-child-radius}
\end{figure}

\section{Timing}
\label{app:timing}

Here we provide errors and runtime statistics for the experiments in the article.
\autoref{tab:singlerefinement} contains the average runtime of a single SDF evaluation for all refinement experiments.
\newchanged{\autoref{tab:sphere-tracing} contains the runtime of sphere tracing on both repaired and pseudo SDF under different number of rays in \autoref{fig:sphere-tracing}.}{new-table-added}
\autoref{tab:runtimes-2dcomparison} contains the statistics for the experiments in \autoref{fig:reconstruction-comparison-2d}, and \autoref{tab:runtimes-3dcomparison} contains the statistics for the experiments in \autoref{fig:reconstruction-comparison-3d}.
\autoref{tab:allfigures} covers most remaining figures.

\section{Further implementation details}
\label{app:furtherimplementation}
We assume that \(\Omega \subseteq [-1,1]^d\), and only consider grids, sphere intersections and uncovered points in \([-1,1]^d\).
We perform all intersection and containment checks in double precision arithmetic.
Our default tolerance for determining uniqueness of points is \(10^{-9}\); all other tolerances are \(10^{-6}\).
Our code is written in Python and C++ using Gpytoolbox \cite{gpytoolbox}, libigl \cite{libigl}, nanoflann \cite{nanoflann} and wgpu \cite{wgpu}.
We run all code on an unbinned Macbook Pro M1 with 16GB RAM, and a Desktop with an Intel i7-13700K with 64GB RAM.

We trained DeepSDF \cite{DeepSDF_CVPR} only on SDF samples on a \(15^3\) grid without near-surface sampling in \autoref{fig:nsdf-violation}.

\newchanged{We use the official C++ implementations released by the authors for RFTA (from gpytoolbox) and both methods proposed by Kohlbrenner. We use the gpytoolbox implementations for Marching Squares/Cubes. All methods are executed using a single CPU thread.
}{new-implementation-details-added}

\section{Parameters}
\label{app:parameters}
By default, we run our algorithm with $\kappa = 4 \cdot n^{\frac{1}{2}}$ for $(d=2)$ and $\kappa = 8 \cdot n^{\frac{1}{3}}$ for $(d = 3)$; and subdivision depth $\subdivisiondepth = 2$ if $(d = 3 $ and $ n> 20^3 )$ else $\subdivisiondepth=3$.
The parameters used in experiments can be found in \autoref{tab:allfigures} and the caption of \autoref{tab:runtimes-scatterplot}.

\begin{table}
    \centering
    \resizebox{\columnwidth}{!}{%
    \begin{tabular}{l|l|l}
    \rowcolor[HTML]{C6DBEF} Shape & Grid & Avg. single SDF evaluation time\\
    \autoref{fig:teaser} , left & $12^3$ & 0.0311 \\

    \autoref{fig:refinement}, top & $14^3$ & 0.0228\\

    \autoref{fig:refinement}, middle & $14^3$ &0.0134\\

    \autoref{fig:refinement}, bottom & $14^3$ &0.0120
    
    \end{tabular}%
    }
    \caption{\changed{Average runtime for a single SDF evaluation for all of the refinement experiments.}{new-table}}
    \label{tab:singlerefinement}
\end{table}

\begin{table}[t]
    \centering
    \resizebox{\columnwidth}{!}{%
        \begin{tabular}{l|l|l|l|l|l}
            \rowcolor[HTML]{C6DBEF}
            SDF &  840M & 1260M & 2100M & 4200M & 5030M\\
            repair time &  1348.582s & 1348.582s &1348.582s & 1348.582s & 1348.582s\\
            repaired SDF & 226.349s & 338.73s & 561.477s &1129.437s & 1531.894s \\
            pseudo SDF  & 407.956s & 610.345s & 1016.285s & 2029.571s &2963.335s
        \end{tabular}%
    }
    \caption{\newchanged{Sphere tracing on the pseudo-SDF takes nearly twice as long as on the repaired SDF. With sufficiently many rays, the accumulated runtime savings exceed the one-time repair cost}{new-table}.}
    \label{tab:sphere-tracing}
\end{table}

\begin{table}
    \centering
    \begin{tabular}{l|l|l|l|l|l|l}
    \rowcolor[HTML]{C6DBEF} Shape & Grid & $\subdivisiondepth$ & $\kappa$ & Chf & Hdf & Runtime\\
      \autoref{fig:teaser}, left  & $12^3$ & 1 & \(\infty\) & - & - & 27.2191  \\
      \autoref{fig:teaser}, middle  & $25^3$ & 2 & 500 & 0.0252& 0.0599& 61.5978 \\
      \autoref{fig:teaser}, middle & $40^3$ & 2 & 500 & 0.0146&0.0364 & 251.2009  \\
      \autoref{fig:teaser}, right & $45^3$ & - & - & - & - & 34326.5295\\
      \autoref{fig:redistancing}, bottom & $20^2$ & 2 & \(\infty\) & 0.0297 &0.0862 & 0.1539\\
      \autoref{fig:radius-scores}, second & $8^2$ & 2 & \(\infty\) & - & - & 0.1227 \\
      \autoref{fig:radius-scores}, third & $8^2$ & 2 & \(\infty\) & - & - & 0.0158\\
      \autoref{fig:radius-scores}, fourth & $8^2$ & 2 & \(\infty\) & - & - & 0.1601\\
      \autoref{fig:subdivision-orders}, middle & $18^2$ & 1 & 100 & - & - & 0.1297\\
      \autoref{fig:subdivision-orders}, right & $18^2$ & 1 & 100 & - &-& 0.0288\\
      \autoref{fig:subdivision-levels}, second & $25^3$ & 0 & 200 &
      0.0216 & 0.0854 & 2.4198 \\
      \autoref{fig:subdivision-levels}, third & $25^3$ & 1 & 200 &
      0.0164 & 0.0837 & 6.2958\\
      \autoref{fig:subdivision-levels}, fourth & $25^3$ & 2 & 200 & 
      0.0153 & 0.0775 & 17.8579\\
      \autoref{fig:subdivision-levels}, fifth & $25^3$ & 3 & 200 &
       0.0150 & 0.0766 &117.0881\\
       \autoref{fig:subdivision-levels}, sixth & $25^3$ & 4 & 200 &
       0.0149 & 0.0669 & 1884.5306\\
      \autoref{fig:naive-vs-smart-upsampling}, left & $15^2$ & 1 & $\infty$ &- &- &0.0118\\
      \autoref{fig:naive-vs-smart-upsampling}, middle & $15^2$ & 1 & $\infty$ &0.0151 & 0.0319 &0.0192\\
      \autoref{fig:naive-vs-smart-upsampling}, right & $15^2$ & 1 & $\infty$ &0.0151 &0.0319 &0.0187\\
       \autoref{fig:refinement}, top & $14^3$ & 2 & 800 & - & - & 60.9490\\
       \autoref{fig:refinement}, middle & $14^3$ & 2 & 800 & - & - &55.3167\\
       \autoref{fig:refinement}, bottom & $14^3$ & 2 & 800 & - & - & 58.2502\\
       \autoref{fig:nsdf-violation}, middle & $15^3$ & 2 & \(\infty\) & - & - & 169.3058\\
       \autoref{fig:pseudo-fmm} & $45^2$ & - & - & - & - & 0.0748\\
       \autoref{fig:nsdf-didactic} & $6^2$ & - & - &- & - & 0.0177\\
       \autoref{fig:offset-3d}, top & $45^3$ & - & - & - & - & 15531.5520\\
       \autoref{fig:offset-3d}, middle & $45^3$ & - & - & - & - & 18959.3387\\
       \autoref{fig:offset-3d}, bottom & $45^3$ & - & - & - & - & 41900.3524\\
       \autoref{fig:comparison-with-mc-interpolation}, second & $25^3$ & 2  & 500 &0.0213 & 0.0909 &40.8223\\
      \autoref{fig:rfta-prior}, middle ours & $35^3$ & 2 & 280 & 0.0136 & 0.0434 & 43.6214\\
      \autoref{fig:rfta-prior}, right ours & $50^3$ & 2 & 400 & 0.0073 & 0.0295 & 214.8491\\
      \autoref{fig:limitations} & $40^3$ & 2 & 400 &0.0156 & 0.0349 &205.7701  \\
      \autoref{fig:noise}, second  & $22^3$ & 2 & 176 & 0.0243 &  0.0595 & 16.2284\\
      \autoref{fig:noise}, third & $22^3$ & 2 & 176 & 0.0252 & 0.0807 & 15.6326\\
      \autoref{fig:noise}, fourth & $22^3$ & 2 & 176 &  0.0309 & 0.0763 & 15.1133\\
      \autoref{fig:noise}, fifth & $22^3$ & 2 & 176 & 0.0420 & 0.0892 & 13.2991\\
      \autoref{fig:noise}, sixth & $22^3$& 2 & 176 & 0.1144 & 0.2306 & 26.3825\\
      \autoref{fig:noise}, seventh & $22^3$ & 2 & 176 & 0.2627 & 0.7076 & 133.1107
      \end{tabular}
    \caption{Parameters, Chamfer errors (Chf), Hausdorff errors (Hdf), and runtime (in seconds) for a variety of experiments in our article.}
    \label{tab:allfigures}
\end{table}

\begin{figure*}
\centering
\includegraphics{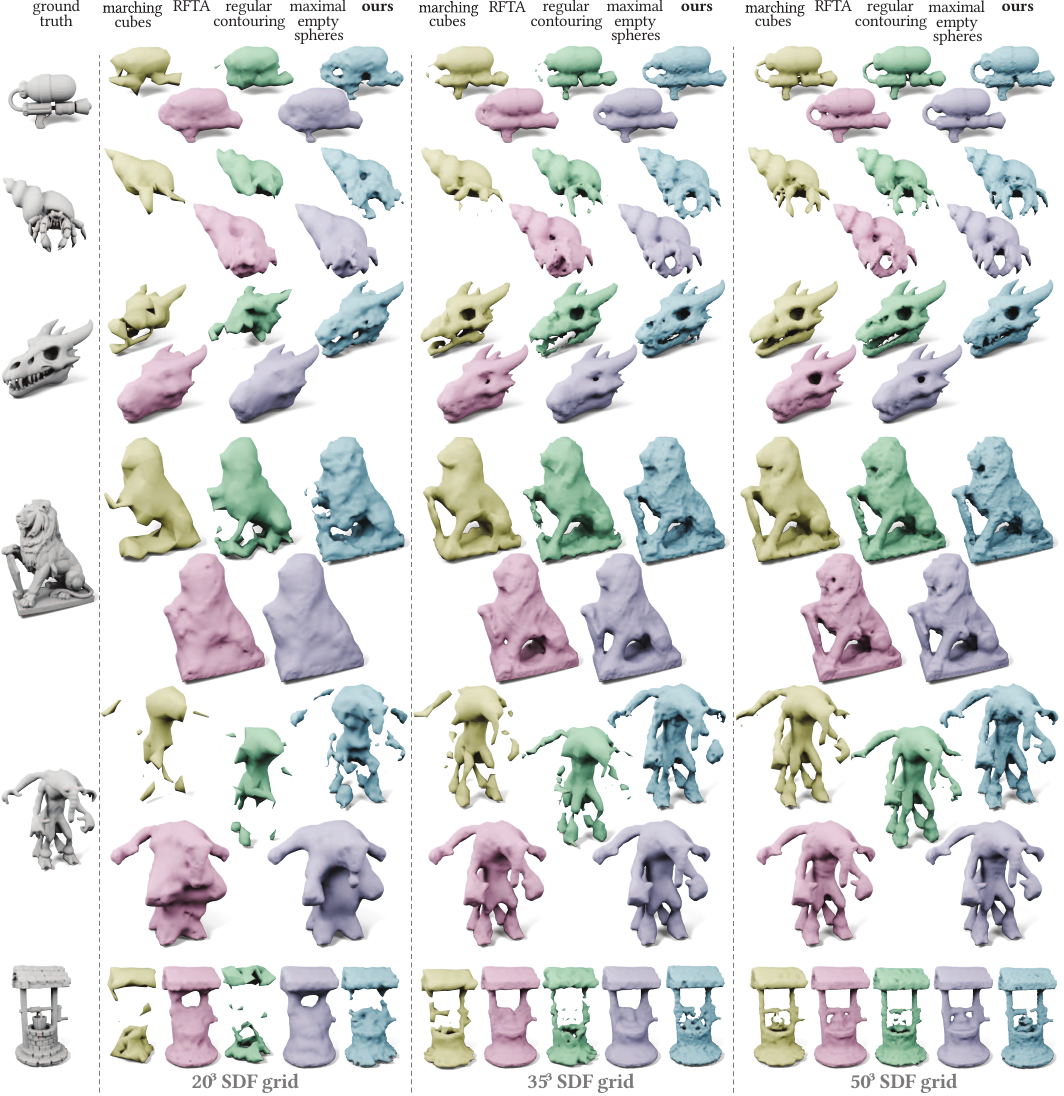}
\vspace{-0.7cm}
\caption{\changed{Additional experiments}{more-experiments} using our method to reconstruct three-dimensional SDFs.
Like other recent methods \cite{rfta,kohlbrenner2025isosurface,kohlbrenner2025polyhedral}, ours yields a detailed surface from low-resolution SDF data, which older local methods \cite{Lorensen1987} cannot do, but our method is even better at extracting intricate details, since it does not contain a strong implicit smoothness prior.}
\label{fig:reconstruction-comparison-3d}
\end{figure*}

\begin{figure*}
\centering
\includegraphics{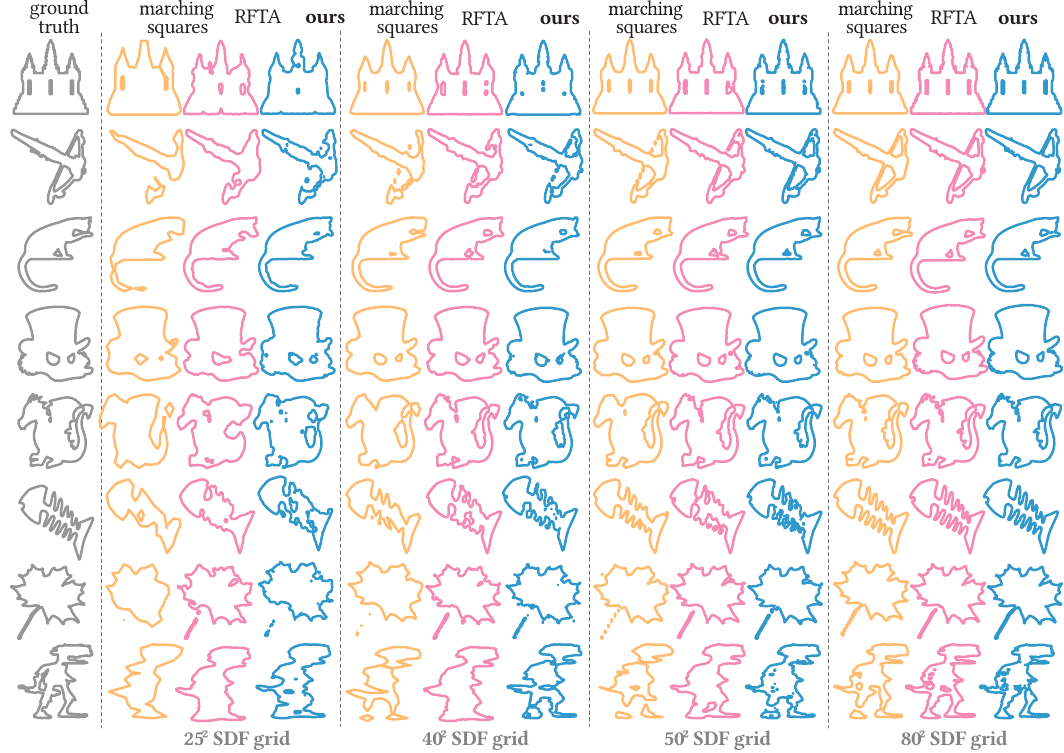}
\caption{\vspace{-2pt}
\changed{We can use our method to reconstruct two-dimensional SDFs with high fidelity.
We use global information, achieving the improved quality of other sphere-based methods \cite{rfta} over local methods \cite{Lorensen1987},
and recover even intricate topological features.}{twodresultsmoved}}
\label{fig:reconstruction-comparison-2d}
\end{figure*}

\begin{table*}
\centering
\resizebox{\textwidth}{!}{%
\begin{tabular}{l|l|l|l|l|l|l|l|l|l|l|l|l|l|l|l}
\rowcolor[HTML]{C6DBEF}Grid & Chf MC & Chf RFTA & Chf RC & Chf MES & Chf ours & Hdf MC & Hdf RFTA & Hdf RC & Hdf MES & Hdf ours & Time MC & Time RFTA & Time RC & Time MES & Time ours
  \\
 $15^3$  & 0.1181         & \textbf{0.0493}      & 0.1232       & 0.0707 & 0.0495 & 0.2771 & \textbf{0.1162} & 0.2894 & 0.1462 & 0.1225 & 0.0006 & 3.8685 & 0.0558 & 1.5023 & 3.2475\\
 $25^3$ & 0.0525 & 0.0339 & 0.0553 & 0.0505 & \textbf{0.0230} & 0.1443 & 0.0889 & 0.1512 & 0.1179 & \textbf{0.0663} & 0.0022 & 21.5992 & 0.1246 & 5.3893 & 17.5931 \\
 $50^3$ & 0.0171 & 0.0196 & 0.0166 & 0.02960 & \textbf{0.0109} & 0.0607 & 0.0563 & 0.0587 & 0.0756 & \textbf{0.0348} & 0.0120 & 242.6865 & 0.6263 & 40.1474 & 367.4088
 
\end{tabular}}
\caption{\changed{%
\vspace{-1pt}
Average errors and runtimes (in seconds) for the experiment in \autoref{fig:big-dataset-scatterplot}.
We report both Chamfer (Chf) and approximated Hausdorff (Hdf) distances to the ground truth for the methods of \citet{Lorensen1987} (MC), \citet{rfta} (RFTA), \citet{kohlbrenner2025isosurface} (RC), \citet{kohlbrenner2025polyhedral} (MES) and ours.
Our method outperforms every other method at medium and high resolutions, and is only narrowly bested by RFTA at low resolutions.
We used the parameters $\kappa = 6 \cdot n^{\frac{1}{3}}, \subdivisiondepth = 2$.
}{runtimesplot-moved}
\label{tab:runtimes-scatterplot}}
\end{table*}

\begin{table*}
\centering
\resizebox{\textwidth}{!}{%
\begin{tabular}{l|l|l l l l l|l l l l l|l}
Shape & Grid & Chf MC & Chf RFTA & Chf RC & Chf MES & Chf ours & Hdf MC & Hdf RFTA & Hdf RC & Hdf MES & Hdf ours &  Time ours
  \\
  \hline
metratron & $20^3$ &  0.1224 & 0.0664 & 0.1236 & 0.0985 & \textbf{0.0553} & 0.2432 & 0.1402 & 0.2477 & 0.1508 & \textbf{0.1011} & 16.2775\\
\hline
metratron & $35^3$ & 0.0326 & 0.0332 & 0.0358 & 0.0654 & \textbf{0.0153} & 0.0960 & 0.0883 & 0.0913 & 0.0995 & \textbf{0.0480} & 100.4788\\
\hline
metratron & $50^3$ & 0.0148 & 0.0099 & 0.0178 & 0.0441 & \textbf{0.0068} & 0.0336 & 0.0395 & 0.0353 & 0.0855 & \textbf{0.0235} & 690.389\\
\hline
rose & $20^3$ & 0.0943 & 0.0793 & 0.0971 & 0.1117 & \textbf{0.0507} & 0.2587 & 0.1662 & 0.2537 & 0.2131 & \textbf{0.1193} & 20.7207\\
\hline
rose & $35^3$ & 0.0335 & 0.0378 & 0.0361 & 0.0609 & \textbf{0.0230} & 0.1203 & 0.0783 & 0.1321 & 0.1211 & \textbf{0.0516} & 123.0539\\
\hline
rose & $50^3$ & 0.0159 & 0.0226 & 0.0191 & 0.0460 & \textbf{0.0128} & 0.0457 & 0.0574 & 0.0530 & 0.0857 & \textbf{0.0411} & 712.1766\\
\hline
loop & $20^3$ & 0.3044 & 0.0252 & 0.3467 & - & \textbf{0.0227} & 0.5644 & 0.0641 & 0.6634 & - & \textbf{0.0569} & 7.1094\\
\hline
loop & $35^3$ & 0.1523 & \textbf{0.0158} & 0.0749 & 0.0218 & 0.0225 & 0.3735 & 0.0531 & 0.1672 & 0.0657 & \textbf{0.0507} & 57.6922\\
\hline
loop & $50^3$ & 0.0481 & \textbf{0.0126} & 0.0498 & 0.0144 & 0.0164 & 0.1193 & \textbf{0.0285} & 0.1187 & 0.04477 & 0.0410 & 338.7764\\
\hline
splatter & $20^3$ & 0.0340 & 0.0430 & 0.0406 & 0.0470 & \textbf{0.0291} & 0.1286 & 0.1116 & 0.1463 & 0.1185 & \textbf{0.0952} & 6.9883\\
\hline
splatter & $35^3$ & 0.0182 & 0.0282 & 0.0170 & 0.0295 & \textbf{0.0161} & 0.05625& 0.0935 & 0.0658 & 0.0937 & \textbf{0.0549} & 50.1380\\
\hline
splatter & $50^3$ & 0.0111 & 0.0202 & \textbf{0.0101} & 0.0191 & 0.0111 & \textbf{0.0498} & 0.0825 & 0.0507 & 0.0779 & 0.0555 & 279.9250\\
\hline
crab & $20^3$ & 0.0474 & 0.0618 & 0.0665 & 0.0625 & \textbf{0.0364} & 0.2190 & 0.1789 & 0.2817 & 0.1923 & \textbf{0.1602} & 8.3974 \\
\hline
crab & $35^3$ & 0.0322 & 0.0315 & 0.0305 & 0.0365 & \textbf{0.0234} & 0.2027 & \textbf{0.0799} & 0.1823 & 0.1138 & 0.1377 & 67.6664\\
\hline
crab & $50^3$ & 0.0243 & 0.0278 & 0.0260 & 0.0272 & \textbf{0.0156} & 0.1835 & 0.1004 & 0.1774 & 0.0914 & \textbf{0.0607} & 410.7270\\
\hline
dragon & $20^3$ & 0.0526 & 0.0732 & 0.0698 & 0.0819 & \textbf{0.0388} & 0.1578 & 0.1830 & 0.2108 & 0.1906 & \textbf{0.0969} & 10.1900\\
\hline
dragon & $35^3$ & 0.0239 & 0.0384 & 0.0244 & 0.0437 & \textbf{0.01876} & 0.0969 & 0.1092 & 0.1125 & 0.1091 & 0.0522 & 69.0681\\
\hline
dragon & $50^3$ & 0.0135 & 0.0303 & 0.0128 & 0.0352 & \textbf{0.0114} & 0.0662 & 0.1021 & 0.0528 & 0.0938 & \textbf{0.0377} & 453.166\\
\hline
lion & $20^3$ & 0.0449 & 0.0591 & 0.0434 & 0.0706 & \textbf{0.0335} & 0.1614 & 0.1693 & 0.1114 & 0.1888 & \textbf{0.1082} & 11.1054\\
\hline
lion & $35^3$ & \textbf{0.0171} & 0.0290 & 0.0181 & 0.0338 & 0.0174 & 0.0577 & 0.0946 & \textbf{0.0564} & 0.0956 & 0.0683 & 77.8197\\
\hline
lion & $50^3$ & \textbf{0.0108} & 0.0159 & 0.0101 & 0.0233 & 0.0111 & \textbf{0.0343} & 0.0742 & 0.0380 & 0.0885 & 0.0448 & 533.4602\\
\hline
monster & $20^3$ & 0.0945 & 0.0708 & 0.0808 & 0.0824 & \textbf{0.0379} & 0.3055 & 0.1560 & 0.2326 & 0.1549 & \textbf{0.1082} & 9.8355\\
\hline
monster & $35^3$ & 0.0290 & 0.0265 & 0.0279 & 0.0336 & \textbf{0.0154} & 0.1089 & 0.0721 & 0.1001 & 0.0985 & \textbf{0.0589} & 58.8823\\
\hline
monster & $50^3$ & 0.0150 & 0.0153 & 0.0143 & 0.0192 & \textbf{0.0085} & 0.0905 & 0.0496 & 0.0753 & 0.0542 & \textbf{0.0363} & 387.6927\\
\hline
well & $20^3$ & 0.05488 & 0.0633 & 0.0634 & 0.0911 & \textbf{0.0347} & 0.1670 & 0.1500 & 0.1726 & 0.1728 & \textbf{0.0832} & 10.5913\\
\hline
well & $35^3$ & 0.0230 & 0.0312 & 0.0249 & 0.0407 & \textbf{0.0198} & 0.0866 & 0.0809 & 0.1012 & 0.0898 & \textbf{0.0517} & 72.5144\\
\hline
well & $50^3$ & 0.0140 & 0.0209 & 0.01377 & 0.0308 & \textbf{0.0132} & 0.0490 & 0.0602 & 0.04655 & 0.0794 & \textbf{0.0447} & 503.4369\\
\end{tabular}}

\caption{\changed{Errors and runtimes (in seconds) for the experiments in \newchanged{\autoref{fig:complicated-geometry}}{new-data} and \autoref{fig:reconstruction-comparison-3d}.}{newexperiments}
We report both Chamfer (Chf) and approximated Hausdorff (Hdf) distances to the ground truth for the methods of \citet{Lorensen1987} (MC), \citet{rfta} (RFTA), \citet{kohlbrenner2025isosurface} (RC), and \citet{kohlbrenner2025polyhedral} (MES) along with ours.
The experiment was run with default parameters.}
\label{tab:runtimes-3dcomparison}
\vspace{-20pt}
\end{table*}

\begin{table*}
\centering
\resizebox{\textwidth}{!}{%
\begin{tabular}{l|l|l l l| l l l| l| l}
Shape & Grid & Chf MS & Chf RFTA & Chf GFTS &Hdf MS & Hdf RFTA & Hdf GFTS& Time RFTA & Time GFTS\\
\hline
angkor & $25^2$ & 0.0459 & \textbf{0.0364} & 0.0487 & 0.1843 & \textbf{0.1024} & 0.1279 & 0.2375 & 0.0571\\
\hline
angkor & $40^2$ & 0.0139 & \textbf{0.0116} & 0.0221 & 0.0365 & \textbf{0.0283} & 0.0948 & 0.4584 & 
0.1802 \\
\hline
angkor & $50^2$ & 0.0108 & 0.0104 & \textbf{0.0091} & 0.0245 & 0.0341 & \textbf{0.0233} & 0.6938 & 
0.2952 \\
\hline
angkor & $80^2$ & 0.0082 & \textbf{0.0055} & 0.0065 & 0.0209 & \textbf{0.0164} & 0.0231 & 3.0581 & 
1.5495 \\
\hline
armbrust & $25^2$ & 0.0342 & 0.0361 & \textbf{0.0249} & 0.0963 & 0.1112 & \textbf{0.0541} & 0.2351 & 0.0539\\
\hline
armbrust & $40^2$ & 0.0196 & 0.0185 & \textbf{0.0142} & 0.0700 & 0.0589 & \textbf{0.0378} & 0.4183 & 0.1729\\
\hline
armbrust & $50^2$ & 0.0121 & 0.0115 & \textbf{0.0096} & 0.0360 & 0.0345 & \textbf{0.0302} & 0.6004 & 0.3141\\
\hline
armbrust & $80^2$ & 0.0080 & 0.0061 & \textbf{0.0051} & 0.0365 & 0.0303 & \textbf{0.0236} & 2.4443 & 1.2559\\
\hline
cat & $25^2$ & 0.0315 & 0.0206 & \textbf{0.0201} & 0.0729 & \textbf{0.0609} & 0.0714 & 0.2272 & 0.0563\\
\hline
cat & $40^2$ & 0.0141 & \textbf{0.0103} & 0.0154 & 0.0574 & \textbf{0.0417} & 0.0521 & 0.3816 & 0.1679\\
\hline
cat & $50^2$ & 0.0111 & \textbf{0.0084} & \textbf{0.0084} & 0.0484 & 0.0388 & \textbf{0.0382} & 0.5992 & 0.3076\\
\hline
cat & $80^2$ & 0.0043 & 0.0056 & \textbf{0.0032} & 0.0217 & 0.0326 & \textbf{0.01615} & 2.4202 & 1.4597\\
\hline
duck & $25^2$ & 0.0238 & 0.0257 & \textbf{0.0173} & 0.0476 & 0.0887 & \textbf{0.0520} & 0.2399 & 0.0529\\
\hline
duck & $40^2$ & 0.0123 & \textbf{0.0099} & 0.0011 & 0.0311 & \textbf{0.0290} & 0.0341 & 0.3850 & 0.1380\\
\hline
duck & $50^2$ &  0.0091 & \textbf{0.0082} & 0.0094 & 0.0322 & \textbf{0.0298} & 0.0478 & 0.6318 & 0.2541\\
\hline
duck & $80^2$ & 0.0067 & \textbf{0.0038} & 0.0049 & 0.0286 & \textbf{0.0162} & 0.0203 & 2.3076 & 1.0027\\
\hline
dragon & $25^2$ & 0.0440 & 0.0369 & \textbf{0.0277} & 0.1682 & 0.1005 & \textbf{0.0589} & 0.2828 & 0.0812\\
\hline
dragon & $40^2$ & 0.0238 & 0.0162 & \textbf{0.0153} & 0.1028 & \textbf{0.0514} & 0.0554 & 0.3917 & 0.1754\\
\hline
dragon & $50^2$ & 0.0151 & 0.0143 & \textbf{0.0119} & 0.0491 & 0.0511 & \textbf{0.0397} & 0.6811 & 0.2855\\
\hline
dragon & $80^2$ & 0.0092 & \textbf{0.0079} & 0.0086 & 0.0481 & 0.0440 & \textbf{0.0402} & 2.7149 & 1.3041\\
\hline
fishbone & $25^2$ & 0.0398 & 0.0544 & \textbf{0.0346} & 0.0760 & 0.1098 & \textbf{0.0076} & 0.2652 & 0.0621\\
\hline
fishbone & $40^2$ & 0.0230 & 0.0265 & \textbf{0.0210} & 0.0542 & 0.0742 & \textbf{0.0531} & 0.4593 & 0.1991\\
\hline
fishbone & $50^2$ & 0.0185 & 0.0233 & \textbf{0.0155} & 0.0586 & 0.0657 & \textbf{0.0466} & 0.6179 & 0.3245\\
\hline 
fishbone & $80^2$ & 0.0060 & 0.0035 & \textbf{0.0033} & 0.0206 & 0.0139 & \textbf{0.0131} & 2.8065 & 1.5737\\
\hline
maple & $25^2$ & 0.0914 & \textbf{0.0246} & 0.0355 & 0.3431 & \textbf{0.0555} & 0.1208 & 0.2301 & 0.0592\\
\hline
maple & $40^2$ & 0.0284 & \textbf{0.0151} & 0.0155 & 0.0892 & 0.0616 & \textbf{0.529} & 0.4112 & 0.1404\\
\hline
maple & $50^2$ & 0.0163 & \textbf{0.0109} & 0.0123 & 0.0569 & \textbf{0.0547} & 0.0548 & 0.6680 & 0.2739\\
\hline
maple & $80^2$ &  0.0081 & 0.0040 & \textbf{0.0028} & 0.0396 & 0.0286 & \textbf{0.0097} & 3.5084 & 1.3598\\
\hline
zombie & $25^2$ & 0.0386 & 0.0355 & \textbf{0.0284} & 0.0716 & 0.0718 & \textbf{0.0635} & 0.2370 & 0.0561\\
\hline
zombie & $40^2$ & 0.0238 & 0.0324 & \textbf{0.0226} & \textbf{0.0457} & 0.0831 & 0.0488 & 0.3729 & 0.1663\\
\hline
zombie & $50^2$ & 0.0215 & 0.0312 & \textbf{0.0204} & \textbf{0.0451} & 0.0668 & 0.0477& 0.6433 & 0.3317\\
\hline
zombie & $80^2$ & 0.0097 & \textbf{0.0081} & \textbf{0.0081} & 0.0371 & 0.0262 & \textbf{0.0260} & 2.7975 & 1.5748\\
\end{tabular}}
\caption{Errors and runtimes (in seconds) for the experiments in \autoref{fig:reconstruction-comparison-2d}.
We report both Chamfer (Chf) and approximated Hausdorff (Hdf) distances to the ground truth for the methods of \citet{Lorensen1987} (MS), \citet{rfta} (RFTA), and ours.
The experiment was run with default parameters.}
\label{tab:runtimes-2dcomparison}
\vspace{-20pt}
\end{table*}

\end{document}